\documentclass[11pt]{article}

\usepackage[english]{babel}
\usepackage[utf8]{inputenc}

\usepackage{fullpage}
\usepackage{amsmath,amsthm,amssymb,amsfonts}
\usepackage{tikz}
\usepackage{graphicx}
\usepackage{dsfont}
\usepackage{todonotes}
\usepackage[colorlinks=true, allcolors=blue]{hyperref}
\usepackage[capitalise]{cleveref}

\newcommand{\remove}[1]{}

\newtheorem{theorem}{Theorem}[section]
\newtheorem{corollary}[theorem]{Corollary}
\newtheorem{lemma}[theorem]{Lemma}
\newtheorem{claim}[theorem]{Claim}
\theoremstyle{definition}
\newtheorem{definition}[theorem]{Definition}
\theoremstyle{remark}
\newtheorem{remark}[theorem]{Remark}

\newcommand{\ceil}[1]{\left\lceil #1 \right\rceil}
\newcommand{\floor}[1]{\left\lfloor #1 \right\rfloor}
\newcommand{\given}{\,\middle\vert\,}
\newcommand{\suchthat}{\,\colon\,}

\newcommand{\dec}{{\cal D}}

\title{Adjacency Sketches in Adversarial Environments}

\author{Moni Naor\thanks{Department of Computer Science and Applied Math, Weizmann Institute of Science; Email: {\tt moni.naor@weizmann.ac.il}. Incumbent of the Judith Kleeman Professorial Chair. Research supported in part by grants from the Israel Science Foundation (no. 2686/20), by the Simons Foundation Collaboration on the Theory of Algorithmic Fairness and by the Israeli Council for Higher
Education (CHE) via the Weizmann Data Science Research Center.} \and Eugene Pekel\thanks{Department of Computer Science and Applied Math, Weizmann Institute of Science; Email: {\tt eugene.pekel@weizmann.ac.il}.}}
\newif\ifthesis
\begin{document}

\ifthesis
\includepdf[pages=-,width=\textwidth]{mscthesis_titlepage.pdf}
\else
\maketitle
\fi

\begin{abstract} 
An adjacency sketching or implicit labeling scheme for a family $\cal F$ of graphs is a method that defines for any $n$ vertex $G \in \cal F$ an assignment of labels to each vertex in $G$, so that the labels of two vertices tell you whether or not they are adjacent.
The goal is to come up with labeling schemes that use as few bits as possible to represent the labels.
By using randomness when assigning labels, it is sometimes possible to produce adjacency sketches with much smaller label sizes, but this comes at the cost of  introducing some probability of error. 
Both deterministic and randomized labeling schemes have been extensively studied, as they have applications for distributed data structures and deeper connections to universal graphs and communication complexity.
The main question of interest is which graph families have schemes using short labels, usually $O(\log n)$ in the deterministic case or constant for randomized sketches.

In this work we consider the resilience of probabilistic adjacency sketches against an adversary making adaptive queries to the labels.
This differs from the previously analyzed probabilistic setting which is ``one shot".
We show that in the adaptive adversarial case the size of the labels is tightly related to the maximal degree of the graphs in $\cal F$.
This results in a stronger characterization compared to what is known in the non-adversarial setting.
In more detail, we construct sketches that fail with probability $\varepsilon$ for graphs with maximal degree $d$ using $2d\log (1/\varepsilon)$ bit labels and show that this is roughly the best that can be done for any specific graph of maximal degree $d$, e.g.\ a $d$-ary tree.
\end{abstract}

\section{Introduction}

We study an adversarial model of implicit graph representations. Traditional approaches rely on a global data structure, such as the adjacency matrix, to represent the structure of a graph. In an implicit representation, on the other hand, the edges of the graph are not explicitly stored but are instead defined through a set of rules or functions. This is achieved by associating a label with each vertex such that the adjacency of any two vertices can be inferred from their labels. The benefit of this approach is that the graph structure is encoded solely in the set of labels, making it particularly useful for graph algorithms, data structures, distributed computing, and communication networks.

Previous works in this area provide for many graph families adjacency labeling schemes using $O(\log n)$-bit labels for $n$ vertex graphs, see Kannan, Naor and Rudich~\cite{kannan1992implicat}, as well as Alstrup et al.~\cite{alstrup2015adjacency}. 
In fact, $\Omega(\log n)$ is a lower bound for any graph where the number of possible neighborhoods is not limited, such as a star. %\footnote{When we are talking about a specific graph, rather than a family of graphs, the problem is similar to that of property preserving hashing of Boyle, LaVigne and Vaikuntanathan~\cite{boyle2019adversarially}.}
Since a vertex-induced universal graph for a family of graphs $\mathcal{F}$ can be constructed from an adjacency labeling scheme for $\mathcal{F}$ (see also \cref{def:universal_graph} below), and vice versa,  those results imply the existence of polynomial sized universal graphs for several graph families. Indeed, many of the results related to adjacency labeling are expressed in terms of the properties of the universal graph into which a graph family is embedded; e.g.\ Alstrup, Dahlgaard and Knudsen~\cite{AlstrupDK17} showed that trees have linear-sized universal graphs from the $\log n +O(1)$ labeling scheme and Alon~\cite{alon2017asymptotically} showed that the family of {\em all} graphs on $n$ vertices has a universal graph whose size is $(1+o(1))\cdot 2^{\frac{n-1}{2}}$ which is asymptotically optimal.

The goal of adjacency schemes is to minimize label size.  So it is also natural to consider adjacency sketches which are randomized representations, aiming for obtaining labels of constant size, but at the cost of having some error probability. Natural probabilistic variants of the results for graph families such as trees were proven by Fraigniaud and Korman~\cite{fraigniaud2009randomized} and various other graph families were studied by Harms, Wild and Zamaraev~\cite{harms2022randomized} using a deeper equivalence of those sketches to communication complexity. Generally, the study of both schemes and sketches was strongly motivated by the implicit graph conjecture, which attempted to identify which graph families have labeling schemes using $O(\log n)$ bits (and its probabilistic version regarding which graph families have adjacency sketches using a constant number of bits). Both conjectures were recently refuted by Hatami, Hatami and Hambardzumyan~\cite{hatami2022implicit,hambardzumyan2022counter} and tight characterizations of the label sizes for general graph families are still unknown.

This work emphasizes a different but inherent question related to the probabilistic properties of adjacency sketches. The correctness of such sketches (and randomized data structures in general) is usually analyzed under the assumption that the vertices we choose to query are fixed and the randomness is over the choice of labels for the graph. Inspired by adversarial models for other data structures, such as the work by Naor and Yogev~\cite{naor2015bloom}, we consider a more robust model where each query is made adaptively by an unbounded adversary {\em after learning the outputs of previous queries}. Other works also consider adversarial models specifically in the context of sketching, e.g.\ several parties computing a joint function on massive inputs as studied by Mirnov, Naor and Segev~\cite{mironov2008sketching} and linear sketches for dimension reduction in the work by Hardt and Woodruff~\cite{hardt2013robust}. Note that whenever we have a sequence of queries where the next one may be a function of the results of the previous answers it is hard to justify a non-adaptive analysis. 

We formalize the notion of implicit graph representations which are resilient to adversarial attacks. Our main contribution is a construction of a resilient scheme for graphs with maximal degree $d$ using $O(d)$-bit labels; more specifically (i)  For probability of forgery $\varepsilon\in (0,1)$ the size of the label is $2d\log(1/\varepsilon)$, and  (ii) We prove a corresponding lower bound showing this is asymptotically tight. While the upper bound works for the family of graphs of maximal degree $d$ (no assumption regarding the specific graph), the lower bound is applicable to almost all graphs of maximal degree $d$ (those that have many vertices of degree $\Theta(d)$), that is even if the scheme is aware of the graph that is being labeled.

In other words, we have a tight characterization of the length of the labels of degree $d$ graphs in the adversarial case, something we do not have in general in other probabilistic settings.

\subsection{Definitions and Notations}

For a natural number $n\in\mathbb{N}$ let $\left[n\right]=\left\{1,\ldots,n\right\}$. For a probability distribution $D$, we write $x\sim D$ when $x$ is a random variable sampled from $D$. A finite set can be written instead of a distribution $D$ to represent the uniform distribution over the set. For an event $A$ we write $\mathds{1}_A$ for the indicator of $A$, i.e.\ the random variable which is $1$ when the event occurs and $0$ otherwise.

The (Shannon) entropy of a distribution $D$ over $\{0,1\}^n$ is defined as
$$H[D]=-\sum_{\alpha\in\{0,1\}^n}\mathds{P}_D\left[\alpha\right]\cdot \log\left(\mathds{P}_D\left[\alpha\right]\right)$$
We denote by $H_\infty [D]$ the {\em min entropy} of the distribution $D$ which is the largest $k\in\mathbb{R}$ such that all outcomes occur with probability at most $2^{-k}$, i.e.\ $H_\infty [D] = -\log\left(\max_{\alpha\in\{0,1\}^n}\mathds{P}_D\left[\alpha\right]\right)$.

\begin{definition}
The statistical distance between two distributions $D$ and $D'$ over $\{0,1\}^n$ is
$$
\Delta\left(D,D'\right)
= \frac{1}{2}\sum_{\alpha\in\{0,1\}^n} \left| \mathds{P}_D\left[\alpha\right] - \mathds{P}_{D'}\left[\alpha\right] \right|
= \max_{S\subseteq\{0,1\}^n} \left| \mathds{P}_D\left[S\right] - \mathds{P}_{D'}\left[S\right] \right|
$$
\end{definition}

For an undirected graph $G$ let $V(G)$ be the vertex set of $G$ and let $E(G)$ be the edge set of $G$. The neighborhood of a vertex $x\in V(G)$ is the set of all vertices adjacent to $x$ ($x$'s neighbors) and it is denoted with $N_G (x)$. The degree of $x$ is the number of neighbors of $x$ and is denoted $\deg_G (x)$. When it is clear what is the graph $G$, we simply use $N(x)$ and $\deg (x)$.

For a family $\cal F$ of graphs let ${\cal F}_n$ be the subset of $\cal F$ containing those graphs with $n$ vertices. 

\begin{definition}[Adjacency Labeling]
An adjacency labeling scheme with size $c(n)$ for a family $\mathcal{F}$ of finite graphs 
consists of two algorithms:
(i)  A decoder $\dec_{\mathcal{F}}\colon \{0,1\}^* \times \{0,1\}^* \mapsto \{0,1\}$ and (ii) An encoder outputting for each graph $G\in\mathcal{F}_n$ a function $f_G\colon V\mapsto \{0,1\}^{c(n)}$ labeling the vertices of $G$ with strings of length $c(n)$, so  that for any two vertices $u$ and $v$, given $f_G(u)$ and $f_G(v)$ it is possible to correctly decide the adjacency of $u$ and $v$ in $G$, i.e.
$$\forall u,v\in V(G) \ :\  \dec_{\mathcal{F}}(f_G(u),f_G(v)) = \mathds{1}_{(u,v)\in E(G)}$$
\end{definition}

\begin{definition}[Universal Graph]
\label{def:universal_graph}
A universal graph with size $c(n)$ for a family $\mathcal{F}$ of finite graphs is a sequence $U_\mathcal{F}=(U_n)_{n\in\mathbb{N}}$ such that $\lvert V(U_n)\rvert = c(n)$ and for all $n\in\mathbb{N}$ and $G\in\mathcal{F}_n$, $G$ is a vertex induced subgraph of $U_n$, i.e.\ there exists a mapping $\phi\colon V(G)\mapsto V(U_n)$ satisfying
$$ \forall u,v\in V(G) \ :\  \mathds{1}_{(\phi (u),\phi (v))\in E(U_n)} = \mathds{1}_{(u,v)\in E(G)}$$
\end{definition}

\begin{definition}[Adjacency Sketch]\label{def:adjacency sketch}
An adjacency sketch with size $c(n)$ and error $\delta >0$ for a family $\mathcal{F}$ of finite graphs consists of two algorithms: a randomized encoder $L_G$ and a deterministic decoder $\dec_{\mathcal{F}}\colon\{0,1\}^*\times \{0,1\}^*\mapsto\{0,1\}$. On input $G\in\mathcal{F}_n$, the encoder outputs a function $\ell\colon V(G)\mapsto\{0,1\}^{c(n)}$ from the distribution $L_G$ of labeling functions. The encoder and decoder satisfy the condition that for all $G\in\mathcal{F}$,
$$ \forall u,v\in V(G) \ :\  \mathds{P}_{\ell \sim L_G} \left[\dec_{\mathcal{F}}( \ell(u),\ell(v)) = \mathds{1}_{(u,v)\in E(G)}\right] \geq 1-\delta$$
The probability is over the randomness of the encoder $L_G$  (choice of labels made by selecting $\ell$).
\end{definition}

\begin{definition}[Probabilistic Universal Graph]
A probabilistic universal graph with size $c(n)$ and error $\delta >0$ for a family $\mathcal{F}$ of finite graphs is a sequence of graphs $U_\mathcal{F}=(U_n)_{n\in\mathbb{N}}$ such that $\lvert V(U_n)\rvert = c(n)$ and for all $n\in\mathbb{N}$ and $G\in\mathcal{F}_n$ there is a probability distribution $\Phi_G$ over maps $\phi\colon V(G)\mapsto V(U_n)$ satisfying
$$ \forall u,v\in V(G) \ :\  \mathds{P}_{\phi \sim \Phi_G} \left[\mathds{1}_{(\phi (u),\phi (v))\in E(U_n)} = \mathds{1}_{(u,v)\in E(G)}\right] \geq 1-\delta.$$
The probability is over the randomness of the mapping $\Phi_G$  (the choice of mapping made by selecting $\phi$).
\end{definition}

In Kannan et al.~\cite{kannan1992implicat} it was observed that adjacency labeling schemes and universal graphs are equivalent concepts. Probabilistic labeling schemes for adjacency were studied in \cite{fraigniaud2009randomized}, but our definitions of adjacency sketches and PUGs are based on Harms~\cite{harms2020universal} and \cite{harms2022randomized}. They mirror the concepts of (deterministic) adjacency labeling schemes and universal graphs in the sense that equivalence can be shown by a similar argument: identify the vertices of $U_n$ with binary strings of length $c(n)=\ceil{\log\lvert V(U_n) \rvert}$, the labeling function $\ell\colon V(G)\mapsto\{0,1\}^{c(n)}$ with the map $\phi\colon V(G)\mapsto V(U_n)$ and the decoder $\dec$ with the edge relation in $U_n$.

In later sections, when it is clear from the context, we will drop the $\mathcal{F}$ and simply denote the decoder with $\dec$ or the (probabilistic) universal graph $U$.

Finally, we mention the communication complexity model known as simultaneous messages: two parties are given inputs  \( x \) and \(y\) respectively and should compute some function \( f \) of the inputs \emph{without communicating with each other}. Instead, each one sends a message to a third party (a referee) who calculates \( f(x,y) \) given the messages received. There are two varieties of this model: with shared randomness and with (only) private randomness. Harms, Wild and Zamaraev~\cite{harms2020universal,harms2022randomized} showed a tight relation between the {\em shared} randomness variety and non-adversarial sketches. We explore the connection to the {\em private} randomness case is Section~\ref{section:smmpc}.

\subsubsection{Resilience to Adaptive Adversaries}
\label{sec: adaptive adversaries}
We are now ready to define the main new concept, that of an adjacency  sketching scheme that is resilient to adaptive adversaries. 

\begin{definition}\label{resilience_definition}
Let $\mathcal{S}$ be an adjacency sketch for a graph family $\mathcal{F}$ which only errs on non-edges. Consider the following game against an adversary:
\begin{enumerate}
    \itemsep 0em
    \item The adversary chooses a graph $G\in\mathcal{F}_n$.
    \item Labels are assigned to $G$ by output $\ell$ of the encoder of $\mathcal{S}$.
    \item The adversary can adaptively ask for the labels of vertices $x_1,x_2,\ldots,x_k\in V(G)$. The number of queries $k$ does not need to be decided in advance, but $k\leq n-2$.
    \item The adversary declares that vertices $x_{k+1}$ and $x_{k+2}$ that have not been asked so far are its candidates.\footnote{For simplicity, we assume throughout the paper that $x_{k+1} \neq x_{k+2}$, but all our schemes can handle the case where $x_{k+1}=x_{k+2}$ as well.}
\end{enumerate}
The adversary wins if it is able to find $x_{k+1},x_{k+2} \in V(G)$ with wrong label decoding. We say that the adjacency sketch $\mathcal{S}$ is adversarial resilient with probability of forgery $\varepsilon \geq 0$ if the adversary wins the game with probability at most $\varepsilon$.
\end{definition}

\paragraph{Amplification:}\label{amplification_method}
Suppose that we have a scheme  with some fixed probability of forgery, say  $1/2$, then we can create a scheme with probability of forgery $\varepsilon$ by using several ($k=\lceil \log (1/\varepsilon) \rceil$) such schemes independently and in parallel and  where the new decoder accepts if and only if all the decoders of the old scheme accepts.  To win against this scheme, the adversary must find a non-adjacent pair of vertices that has wrong decoding in all $k$ coordinates. 
Since the original probability of forgery  was $1/2$ and the labeling functions are independently sampled, the new probability of forgery is $(1/2)^k$: the general argument for amplification of the failure of the adversary follows that of parallel repetition of (single prover) interactive proofs (see Goldreich~\cite{Goldreich98} Appendix C.1). 

\bigskip
\noindent
The general question we explore is how small can the label size be as a function of local properties of the graph and the desired error (or probability of forgery).  We will mostly discuss adjacency sketches which only err on non-edges. Unless explicitly stated otherwise, the term ``labeling scheme" or ``scheme" will be used to refer to such constructions.

\subsection{Our Results}
The work presented in this paper characterizes the length of adjacency sketches in the adversarial model introduced in \cref{resilience_definition}.  Recall that our labeling schemes refer to adjacency sketches which only err on non-edges and where the adversary  adaptively queries the labels of vertices and attempts to find a pair of {\em non-adjacent} vertices which are decoded as having an edge according to the scheme. The discussion in \cref{section:discussion} motivates the choice to focus on this type of sketches and explores some variations of the adversarial model.

We start with an investigation of resilient labeling schemes in the simplified scenario of graphs where the degree is at most 1, i.e.\ every vertex has at most one neighbor - a matching (\cref{section:low_degree}).
A reasonable approach is to look at pairs of adjacent vertices, give one of them a random label $u$ from the set of possible labels and give the other vertex a random label $v$ from the subset of labels $u$ accepts. The main question is how to choose the structure, i.e.\ the subset of labels each $u$ accepts, such that this is optimal. For example, coloring-based schemes use colors as labels and every color accepts all colors excluding itself (exactly like proper coloring of graphs in general). The main result of this section is the construction of \cref{projective_plane_scheme} based on finite projective planes which achieves probability of forgery $\varepsilon$ using labels of size roughly $2\log (1/\varepsilon)$.

We derive a corresponding lower bound as a function  of $\varepsilon$ by analyzing natural adversarial strategies in terms of the distribution a labeling scheme induces on the corresponding probabilistic universal graph (PUG). This can be summarized as follows:

\begin{theorem}
\label{thm: projective}
For any $\varepsilon\in (0,1)$ there is a labeling scheme (following \cref{resilience_definition}) for the matching graph with labels of size roughly $2\log (1/\varepsilon)$ and probability of forgery $\varepsilon$. Furthermore, this is the best possible label size.
\end{theorem}

The focus of \cref{section:high_degree,section:learning} is bounded degree graphs where the maximal degree is $d$. The goal is to derive a tight characterization of the labeling schemes for the family of such graphs.

\paragraph{Upper Bound:}
First, we generalize some of the earlier coloring-based examples in the paper and show that even with relatively small labels one can, somewhat surprisingly, achieve non-trivial resilience against an adversary.

\begin{theorem}
\label{thm: coloring}
There is a labeling scheme (following \cref{resilience_definition}) for the family of graphs with maximal degree $d$ that has a probability of forgery at most $1-\Omega (1/d)$ using labels of size $O(\log d)$.
\end{theorem}

This theorem also naturally implies, via amplification, the existence of a labeling scheme using labels of size $O(d\log d)$ having a constant probability of forgery. In \cref{covering_by_matchings} we show that this can be improved. We start from the projective plane method for matching labeling and cover the labeled graph using at most $d+1$ matching (this is always possible by Vizing's theorem), effectively using $d+1$ projective plane schemes simultaneously. The size after this first step would still be $O(d\log d)$. 

The technique employed to reduce the size to $O(d)$ bits is to use a data structure for retrieval.
The key idea is the following: if we think of the label of each vertex as a collection of smaller labels (``sub-labels")), our setting doesn't require us to explicitly store for every vertex which matching has which smaller label. We only need to retrieve the correct smaller label when decoding. This is similar to the retrieval problem, where we have a subset $S$ of keys from a large universe and each key is associated with a string. The requirement is to answer retrieval queries with the correct string for keys in $S$ and arbitrary string for keys not in $S$. We prove that the construction of a retrieval data structure by Dietzfelbinger and Pagh~\cite{dietzfelbinger2008succinct} also satisfies an additional property that is very useful for    our setting and construction: we want the output of queries retrieving strings for keys {\em not} in $S$ be {\em unpredictable} (the outcome should have high entropy from the adversary's point of view). We get the following:

\begin{theorem}
\label{thm: general}
For any $\varepsilon\in (0,1)$ there is a labeling scheme (following \cref{resilience_definition}) for the family of graphs with maximal degree $d$ using labels of size roughly $2d\log (1/\varepsilon)$ and probability of forgery $\varepsilon$. The scheme is computationally efficient: there is a polynomial time algorithm that given a graph finds the appropriate labeling. 
\end{theorem}

If we interpret this result as regarding probabilistic embedding, then it says that the family of max degree $d$ graphs has a probabilistic universal graph of size $O((1/\varepsilon)^{2d})$ where any graph from the family  can be embedded and a faulty pair cannot be found, even in an  adaptive manner, with probability better than $\varepsilon$.

Also, note that the size of this construction is within a constant of the size of the construction for the non-adversarial setting (see \cref{apdx:degree_d_non_adversarial}).

\paragraph{Lower Bound:}
To conclude the characterization of schemes for graphs with maximum degree $d$, we show a linear in $d$ lower bound (in \cref{section:learning}).
Our major finding is that the aforementioned bound is tight not only w.r.t.\ the {\em family} of graphs of maximum degree $d$, but also holds true for any {\em specific} graph, provided that there are a sufficient number of vertices of degree $d$, e.g.\ a $d$-ary tree.

We introduce an adversarial strategy based on the algorithm for learning {\em adaptively changing distributions} (ACDs) by Naor and Rothblum~\cite{naor2006learning}. The general task is similar to impersonating one party in a conversation where two legitimate parties share a secret 
and the objective is to identify impersonators by observing the public messages exchanged between them.
The ACD is the distribution over the next message, and the objective is to produce a distribution that is close to the real one without access to the secret. In our setting, we can think of the secret as the label of a vertex $x$ which the adversary has not seen, and the messages are the labels given to the neighbors $\{\gamma^x_1,\ldots,\gamma^x_d\}\subseteq N(x)$. The main question is how close can the adversary get to the distribution over the labels for $\gamma^x_i$ by observing the labels of $\gamma^x_1,\ldots,\gamma^x_{i-1}$. A key insight by \cite{naor2006learning} which is also used in the variant described in \cref{section:lower_bound_strategy} is that we expect to either predict the label of $\gamma^x_i$ (in the sense of having a close distribution to it) or to learn significant information about the label of $x$. This is the case since when the next neighbor label is hard to predict, the adversary gains a lot of information about the label of $x$ by querying the next neighbor.

When the maximal degree of the graph is large enough, as compared to the label size, then the adversary is able to execute enough rounds of the algorithm to effectively learn the distribution over labels of a vertex that wasn't queried and select candidates for forgery accordingly. This leads to the following result:

\begin{theorem}
\label{thm: lower}
Let $\alpha\in (0,1)$. For every $d\in\mathbb{N}$, there exists some $n>0$ such that for any graph containing $n$ non-adjacent vertices of degree at least $d$ and any labeling of that graph (following \cref{resilience_definition}) using labels of size less than $\Omega(d\cdot\alpha^4)$ bits, the probability of forgery by an adversary is at least $1-\alpha$.
\end{theorem}

The result actually refers to any particular (large enough) graph in the family of graphs with maximal degree $d$ and not just to the family of graphs with maximal degree $d$. Since this family is relatively large, it is not surprising that labels of at least $d$ bit are required even in the non-adversarial setting (see \cref{apdx:degree_d_non_adversarial}). Our claim is stronger in the sense that even if the graph is known in advance, and even if it has a particular structure such as a tree, labeling it using short labels (compared to $d$) will result in a probability of forgery that is close to 1. That is as long as the graph contains enough non-adjacent vertices of degree $d$.

A motivating example for this lower bound would be a tree with degree $d$. It is known that in the non-adversarial setting one can label it using labels of constant size (independent of $d$), but as we show this is impossible in the adaptive adversarial setting.

\paragraph{How large should \texorpdfstring{$n$}{n} be?} The proof of Theorem~\ref{thm: lower} actually yields a bound on $n$ which is double exponential in the degree $d$, so for larger $d$, say as in the hypercube (a specific graph on $2^d$ vertices) this results in a $\Omega(\log d)$ lower bound for the label size in the adversarial setting. In \cref{section:smmpc} we improve the analysis to only require $2^{O(d^2)}$ vertices of high degree. This is done using a technique from the lower bound on the communication complexity of computing equality in the simultaneous messages model with private coins (SMMPC) by Babai and Kimmel~\cite{BabaiK97}. This yields a $\Omega(\sqrt{d})$ lower bound on the label size for the hypercube.

Furthermore, we note that the SMMPC equality protocol can also be used to construct adversarial resilient schemes using labels of size $O(\sqrt{\log n})$ (where $n$ is the size of the graph) for certain graphs. This proves useful for trees where $d$ is large, say $\Theta(\log n)$, and for the hypercube as well. In fact, the construction for the hypercube shows that the previously mentioned lower bound is in fact tight.

\subsection{A Discussion of The Model}\label{section:discussion}
We consider some variants of the adversarial model described above and see the consequences.

\paragraph{The Single Vertex Guess Model:}
A possible definition of resilience against an adaptive adversary is where one of the vertices of the candidate pair was queried before, but the other one was not. In this model, the adversary asks for the labels $x_1, x_2, \ldots, x_k\in V(G)$ and then only has to guess a vertex $x_{k+1}$ that has not appeared so far such that the decoding for the pair $\ell (x_k), \ell (x_{k+1})$ is wrong (the pairing with $x_k$ is w.l.o.g. as the adversary can ask about any vertex again). Again, we consider a one-sided error of the type where $(x_k,x_{k+1})\notin E(G)$ but the decoding of their labels is 1. This definition is impossible to widely satisfy due to the following strategy for the adversary:

\begin{claim}
Consider some labeling scheme using less than $m$ labels and suppose a graph $G$ containing an induced matching on $2m$ vertices is labeled using the scheme. Then, the adversary in the single vertex model can win against such a labeling of $G$ with probability 1.
\end{claim}

\begin{proof}
Let $G$ be the labeled graph with the induced matching on $2m$ vertices whose edges are $(x_i,y_i)\in E(G)$ for $i\in\left[m\right]$. The adversarial strategy is to query the labels of $x_1,\ldots,x_m$. Since there are less than $m$ labels, there is a pair $x_i$ and $x_j$ such that $i\neq j$ but $\ell(x_i)=\ell(x_j)$. The adversary will select as candidates $x_i$ and $y_j$.

Assuming the scheme only errs on non-edges, $\ell(x_j)$ must accept $\ell(y_j)$ but $\ell(x_i)=\ell(x_j)$ and $x_i$ is not adjacent to $y_j$. This means the adversary is guaranteed to win.
\end{proof}

\paragraph{Black Box Model:}
A different model of the adversary is inspired by the notion of an evaluation oracle introduced by Boyle, LaVigne and Vaikuntanathan~\cite{boyle2019adversarially} (called  there Evaluation-Oracle PPH\footnote{Boyle, LaVigne and Vaikuntanathan~\cite{boyle2019adversarially} as well as other, such as Holmgren et al.~\cite{HolmgrenLTW22}, consider Property Preserving Hashing which is related to sketching, but assumes a particular function or graph, rather than a family, and the label or sketch is the value of a hash function applied to the name. Their notion of 'robust' is related to our notion of adversarial resilience. They suggest a hierarchy of adversarial access to the hash function assigning the label.}). This adversary only has ``black-box" access to the labeling scheme, in the sense that the results of queries are outputs of the decoder for the labeling scheme and the adversary never observes any labels directly. More formally, the queries performed by the adversary are of the form $(x,y)\in V(G)\times V(G)$ and as a result it learns $\dec(\ell(x),\ell(y))$. The adversary can ask about any number of pairs. Then, it has to pick a new pair $(x',y')\notin E(G)$ and the adversary wins if
$\dec(\ell (x'),\ell (y'))=1$. Note that it is possible the vertices $x'$ and $y'$ participated in past queries. The only restriction is that the adversary did not query $(x',y')$ as a pair.

\begin{claim}
Let $\varepsilon\in \left(0,1\right)$. Consider some labeling scheme using $m$ labels and suppose a graph $G$ containing an induced matching on at least $\frac{4m^2}{\varepsilon}$ vertices is labeled using the scheme. Then, the adversary in the black box model can win against such a labeling of $G$ with probability at least $1-\varepsilon$.
\end{claim}

\begin{proof}
Let $G$ be the labeled graph containing an induced matching on $2n$ vertices whose edges are $(x_i,y_i)\in E(G)$ for $i\in\left[n\right]$. Assume that $G$ is labeled using $m$ labels. We suggest an adversarial strategy based the following observation: every label is characterized by its relation to other labels, i.e.\ we can determine the label of a vertex (up to equivalence) by querying it with vertices having all possible labels.

The adversarial strategy is to first select a random edge $(x_k,y_k)\in E(G)$ and then query all pairs $(x_i,x_j),(x_i,y_j),(y_i,y_j)$ for $i,j\neq k$. The queries performed in this step effectively split all vertices in the matching (except $x_k$ and $y_k$) into sets $V_1,\ldots,V_m$. The vertices in each set $V_i$ either have the same label or equivalent labels (in the sense that they accept or reject the same vertices in the matching). Note that it could be the case that there are additional labels used outside the matching, i.e.\ there are empty sets among $V_1,\ldots,V_m$, but this is irrelevant unless those labels are used for either $x_k$ or $y_k.$

Assume there is some $k'\neq k$ such that $(x_k,y_k)$ and $(x_{k'},y_{k'})$ are labeled by the same pair of labels. In the next step, the adversary queries $x_k$ together with all vertices in the matching except $y_k$. Let $V_i$ be the set whose vertices behave the same as $x_k$ in relation to all other vertices in the matching. By our assumption, $V_i$ is not empty and w.l.o.g. $x_{k'}\in V_i$. Then, the candidates selected by the adversary are $y_k$ and $x_{k'}$. Since the scheme only errs on non-edges and $\ell(x_y)$ accepts $\ell(y_k)$, so will $\ell(x_{k'})$. As a result, the adversary wins.

It remains to show that the assumption about $k'$ holds with probability $1-\varepsilon$ for large enough $n$. We say the edge $k$ in the matching is common if there are at least $\frac{\varepsilon n}{m^2}$ edges in the matching using the same pair of labels. By the union bound on pairs of labels, our randomly chosen edge $k$ is uncommon with probability at most $\varepsilon$. When it is common, the edge $k'$ with the same pair of labels exists for $n\geq\frac{2m^2}{\varepsilon}$.
\end{proof}

\paragraph{Does the adversary know it is about to win?}
A stronger notion of adversarial resilience is the ``Bet-or-Pass" model described by Naor and Oved~\cite{naor2022bet}. The main difference between this model and the game described in \cref{resilience_definition} is that the adversary can either output a pair of candidates or choose to forfeit the current iteration of the game. When it goes forward it should have a probability no better than $\varepsilon$ of succeeding and the adversary should decide to bet with a probability greater than $0$. Our main constructions maintain their resilience against a ``Bet-or-Pass" adversary.

\paragraph{What side is the error on?}
We provide some motivation for the choice to focus only on adjacency sketches that err on non-edges, i.e.\ for every edge $(x,y)\in E(G)$ in the labeled graph $G$ the decoder must satisfy $\dec(\ell(x),\ell(y))=1$ and errors only occur for pairs of vertices which are not adjacent.

As mentioned above, when considering one-sided error, it is known that randomness can be used to improve label size compared to the deterministic setting for sketches which only err on non-edges. One example is adjacency labeling for trees which requires $\Omega (\log n)$-bit labels in the deterministic case but can be done using constant size labels in the probabilistic setting as seen in \cite{fraigniaud2009randomized}. On the other hand they also showed that this is not the case for probabilistic sketches which only err on edges, making the study of the adversarial resilience of such sketches irrelevant.

\begin{theorem}[\cite{fraigniaud2009randomized}]
Let $\varepsilon \in (0,1)$. Any adjacency sketch for trees that only errs on edges with probability at most $\varepsilon$ must use labels of size $\log n+\log (1-\varepsilon)-O(1)$ for $n$-vertex trees.
\end{theorem}

A possible relaxation of Definition~\ref{resilience_definition} is to allow two-sided errors, i.e.\ to err both on edges and non-edges (with an error probability of $1/3$ or some other fixed value bounded away from $1/2$). All the constructions (Theorems ~\ref{thm: projective}, \ref{thm: coloring}, and \ref{thm: general}) hold of course, and the lower bound (Theorem~\ref{thm: lower}) can be modified for the two-sided variant  as discussed in~\cref{lower bound for two-sided error}. We choose to focus on the one-sided variant due to the natural one-sided properties of our constructions and for the simplicity of some of the proofs.

\paragraph{Amplification with Two-Sided Error:}
The natural amplification of error technique is to employ several schemes independently in parallel and take the majority. We can get a probability of forgery $\varepsilon$ by using $k=O(\log (1/\varepsilon ))$ schemes with some fixed probability of forgery bounded away from $1/2$, but unlike \cref{amplification_method}, accept a pair of labels when at least $k/2$ of the schemes accept.

This technique is described in more detail in Proposition 2.2 of~\cite{harms2022randomized} for sketches in the non-adversarial case. In our case, we need to refer to a general argument for games against an adversary such as parallel repetition of single prover interactive proofs as in~\cref{amplification_method}.

The proof we mentioned in~\cref{sec: adaptive adversaries} (due to Goldreich~\cite{Goldreich98}) bounds the winning probability when using parallel repetition and where the adversary wins only if it wins in all games. This is also true for any subset $S\subseteq \left[k\right]$ of the $k$ repetitions: the probability  the adversary wins in all games in $S$ is the product of the probabilities of winning each game (again from~\cite{Goldreich98}). To apply this result to the case where the final winner is determined by the player who won in a majority of games, we can apply the generalized concentration bound (an extension of the Chernoff-Hoeffding bounds) by Panconesi and Srinivasan~\cite{panconesi1997randomized} (see also Impagliazzo and Kabanets~\cite{impagliazzo2010constructive}). These bounds do not assume independence of the events, just that for every subset, the probability that all of it be '1' should be the product of the individual probabilities.\footnote{For instance, Theorem 3.1 in \cite{impagliazzo2010constructive} says that if $X_1, X_2, \ldots, X_k$ are Boolean random variables and there is a $\delta$ such that for every set $S \subseteq  [k]$, $\mathds{P}[\wedge_{i \in S} X_i = 1] \leq \delta^{|S|}$. Then for any $\gamma$ such that $\delta \leq  \gamma \leq 1$, 
$$\mathds{P}\left[\sum_{i=1}^k X_i > \gamma k\right] \leq  
e^{-k  \cdot D(\gamma\parallel \delta)}
$$ where $D(\cdot\parallel \cdot)$ is the relative entropy function.} 

\section{Low Degree Graphs}\label{section:low_degree}

We present several labeling schemes for low-degree graphs (where each vertex has one or two neighbors). Matchings are of particular interest. The scheme we suggest for labeling a matching using projective planes (see \cref{projective_plane_scheme}) will also be used later on for our main positive results regarding graphs with maximal degree $d$. We also provide a tight lower bound for matching labeling (Theorem~\ref{matching_lower_bound} and Corollary~\ref{cor: tightness projective}).

\subsection{Coloring Paths}\label{coloring_paths}

The first schemes we examine have very simple decoding: the labels can be thought of as `colors'  and labels for two vertices are decoded as adjacent iff they have different colors. 
We will show how to use colors so as to label a collection of paths such that the probability of finding a wrong decoding (forgery by an adversary) is bounded by a constant, specifically $7/8$. The idea is to choose an arbitrary endpoint and color the path in a greedy and random manner with 6 colors starting at that endpoint. We continue in the same direction to its neighbor and color it using a random color different from the previous color and so on.

%In this scheme, labels for two vertices are decoded as adjacent if they have different colors. 
We want to argue that the probability of two non-adjacent vertices having the same color (given information about the other vertices) is at least some positive constant.

\begin{claim}
Given any configuration of colors generated by this scheme, suppose that we know all the colors other than vertices $x$ and $y$ which aren't neighbors. Then, the probability that the colors of $x$ and $y$ are the same is at least $1/8$.
\end{claim}

\begin{proof}
For a vertex $x$, let $S_x$ be the set of colors {\em not} used by its neighbors. The size of $S_x$ is either $4$ or $5$, depending on whether its two neighbors have the same color or not.  Note that the distribution of the color of $x$ conditioned on the known colors is uniform from the set $S_x$. This is true because the color for $x$ is chosen at random to be different from the previous color and we know it is not the next color. The same also applies to $y$ with the set $S_y$. Since the palette is of size 6 we have $\left| S_x \cap S_y \right| \geq 1$ which means there is at least one color available to both.

If both $S_x$ and $S_y$ have size 4, then the probability that $x$ and $y$ have the same color is $1/8$, if one has size 4 and the other has size 5, then the probability is $3/20$ and if both have size 5, then the probability is $4/25$. In general, for $p>4$ colors the probability is at least $$\min{\left\{\frac{p-4}{(p-2)^2},\frac{p-3}{(p-2)(p-1)},\frac{p-2}{(p-1)^2}\right\}}=\frac{p-4}{(p-2)^2}\leq 1/8$$
with equality when $p=6$.
\end{proof}
 
Since in any color configuration generated by this scheme, the probability of two non-adjacent vertices to have the same color is at least $1/8$, then at any point the adversary decides on two non-adjacent vertices as its `victims' it wins against this scheme with probability at most $7/8$.

\subsection{Labeling a Matching}

The scheme we used for a collection of paths can also be used to label a matching, which is just a collection of single-edge paths. Assuming a palette size $p$, for every vertex $x$ the set $S_x$ has size $p-1$, so two non-adjacent vertices have a different color with probability at least $\frac{p-2}{(p-1)^2}\leq 1/4$ where the best is achieved with $p=3$ colors. Thus, the probability of forgery is at most $3/4$. We would like to use this as a building block in more general schemes so in this section we will try to improve the probability of forgery and investigate the limit of what can be achieved for a matching.

%\subsubsection{Amplification}%\label{amplification_method}
Note that one can take the previous scheme using 3 colors with a probability of forgery $3/4$ and reduce the probability of forgery by using several such labels in parallel as described in \cref{amplification_method}.
The amplified probability of forgery using $k$ copies in parallel is $(3/4)^k$.
This means that for any $\varepsilon\in\left(0,3/4\right)$, we reduce the probability of forgery to $\varepsilon$ by having labels of size $$\ceil{\frac{\log (3) \log (1/\varepsilon)}{\log (4/3)}}\approx\ceil{3.82\log (1/\varepsilon)}.$$
We will see that it is possible to obtain a better constant than $3.82$ and get down to $2$ (which is optimal).

\subsubsection{The Projective Plane Method}\label{projective_plane_scheme}
We suggest an efficient method for labeling a matching where the probability of forgery is roughly $1/\sqrt{m}$, when there are $m$ possible labels. As we will see in the next section, this is the best possible for any labeling scheme for a matching. The construction is similar to the technique used by Gilbert, MacWilliams and Sloane~\cite{gilbert1974codes} in 1974 in order to construct a message authentication code that can detect forgery. An adversary altering a message using that code will succeed with probability at most $1/\sqrt{K}$ where $K$ is the number of possible keys (and this is also the best possible key size for such a code).

The method is based on finite projective planes. Recall that a projective plane is a set of points, a set of lines and an incidence relation such that:
\begin{enumerate}
    \itemsep 0pt
    \item Any two distinct points are incident to a single line.
    \item Any two distinct lines are incident to a single point.
    \item There exists a set of four points such that no line is incident to three of the points.
\end{enumerate}
For a prime $p$ (or a prime power), a finite projective plane of order $p$ is known to exist. A finite projective plane of order $p$ has $p^2+p+1$ points (as well as lines) and $p+1$ points on each line (also lines through each point).

The labeling scheme using a finite projective plane of order $p$ is as follows:  On every edge of the matching $G$ label the endpoints by a random pair of a point and a line such that the point is on the line. Each label has size $\ceil{\log(p^2+p+1)}$ bits which is determined by the number of points and lines. We do not keep an additional bit to distinguish between points and lines because the finite projective plane of order $p$ is self-dual, so both points and lines can be represented by the same encoding. Two labels are decoded as an edge if and only if the point is indeed on the line.

\begin{claim}
\label{claim: projective}
Let $\varepsilon\in \left(0,1\right)$ and let $p$ be the first prime such that $p+1\geq 2/\varepsilon$. The construction using the projective plane of order $p$ yields a scheme with probability of forgery at most $\varepsilon$ and size at most $\ceil{2\log (4/\varepsilon)}$.
\end{claim}

\begin{proof}
Given $\varepsilon$ and $p$, it is clear the size of the labels is
$$\ceil{\log\left(\text{\# of points}\right)}=\ceil{\log(p^2+p+1)}\leq \ceil{2\log (p+1)}\leq \ceil{2\log (4/\varepsilon)}.$$
To analyze the probability of successful forgery, we assume without the loss of generality that the adversary selects $y,w\in V(G)$ as candidates for forgery (i.e.\ $y$ and $w$ are not adjacent) and $(x,y),(z,w)\in E(G)$ are two edges in the matching such that the labels $\ell(x)=P$ and $\ell(z)$ are known to the adversary,
\begin{align*}
\mathds{P}\left[\dec(\ell(y),\ell(w))=1\right] &
= \mathds{P}\left[ \ell(y)\text{ incident to }\ell(w) \right] \\&
= \mathds{P}\left[ \ell(w)=P \wedge \ell(y)\text{ incident to }\ell(w) \right]
+ \mathds{P}\left[ \ell(w)\neq P \wedge \ell(y)\text{ incident to }\ell(w) \right] \\&
= \mathds{P}\left[ \ell(w)=P \right]
+ \mathds{P}\left[ \ell(w)\neq P \wedge \ell(y)\text{ incident to }\ell(w) \right] \\&
= \mathds{P}\left[ \ell(w)=P \right]
+ \frac{1}{p+1} \left( 1-\mathds{P}\left[ \ell(w)=P \right] \right) %\\&
%\leq \frac{1}{p+1} \left( 2-\frac{1}{p+1} \right) \\&
%\leq \frac{2}{p+1}
\end{align*}
We used the fact that by construction $\ell(y)$ must be incident to $\ell(x)=P$ so given $\ell(w)\neq P$, the probability that $\ell(y)$ is incident to $\ell(w)$ is the probability that out of $p+1$ lines incident to $P$, the only line which is incident to both $P$ and $\ell(w)$ was selected as the label of $y$.

A priori, the probability that $\ell(w)=P$ is $\frac{1}{p^2+p+1}$, because $\ell(w)$ is selected randomly. Here $\ell(z)$ is known and the adversary infers that $\ell(w)$ is one of $p+1$ points incident to $\ell(y)$. This means that to maximize the probability to hit $P$ the adversary should find a vertex $y$ such that $P$ is among the points incident to $\ell(y)$.

This means that $\mathds{P}\left[ \ell(w)=P \right]\leq\frac{1}{p+1}$ and hence $\mathds{P}\left[\dec(\ell(y),\ell(w))=1\right]\leq \frac{2}{p+1}\leq\varepsilon$. %\moni{Is this tight? Can you get the probability of forgery to be nearly 2/(p+1) and not 1/((p+1)?}
\end{proof}

\begin{remark}\label{projective_plane_min_entropy}
Note that if the labels are chosen by a certain distribution rather than the uniform one, then the probability of forgery is determined by the {\em min entropy} of the distribution.
For an edge $(x,y)$ in the matching consider the conditional distribution of the line $\ell(y)$ given the point $\ell(x)$ which we do not assume to be uniform. Our main argument was that for any point $P'\neq \ell(x)$, there are $p+1$ lines incident to $\ell(x)$ and only one of them is incident to both $\ell(x)$ and $P'$. The probability of forgery is determined by the probability that $\ell(y)$ is indeed the line incident to both $\ell(x)$ and $P'$ which is at most $2^{-H_{\infty} \left( \ell(y) \given \ell(x) \right)}$. This is true since the min-entropy gives an upper bound on the probability of any outcome in the distribution of $\ell(y)$ given $\ell(x)$.
\end{remark}

\subsubsection{Lower Bound for Matching Labeling}

\begin{theorem}\label{matching_lower_bound}
Given any labeling scheme (satisfying \cref{resilience_definition}) for a matching using $m$ labels, the probability of successful forgery by an adversary is at least $\Omega\left(1/\sqrt{m}\right)$.
\end{theorem}

\begin{proof}
Assume the graph $G$ is a matching labeled by a scheme whose induced PUG is a graph $U$ with $m$ vertices. For a label $u\in V(U)$ let $q_u$ be the probability that a random vertex in $G$ has the label $u$.

Fix some $\alpha\in (0,1)$ and let $U_\alpha$ be the subgraph of $G$ induced by vertices $u\in V(U)$ for which $q_u\geq\frac{\alpha}{m}$. The graph $U_\alpha$ is not empty since there must be at least one label $u$ such that $q_u\geq \frac{1}{m}$. For any $u\in V(U)$, including those not in $U_\alpha$, let $\deg_{U_\alpha}(u)$ be the number of vertices in $U_\alpha$ which are adjacent to $u$ in $U$. Denote by $\Delta_\alpha$ the weighted mean of those vertex degrees, i.e.\ 
$$\Delta_\alpha = \sum_{u\in V(U)}q_u\cdot\deg_{U_\alpha}(u).$$
Consider the following two strategies for the adversary - we will show that at least one of them has the desired probability of successful forgery: 
\begin{enumerate}
    \item Select a random vertex in $G$ and ask for the label of its neighbor. Let $u$ be the resulting label. Repeat this step, discarding the previously used edge of $G$ in every iteration, until $u$ satisfies $\deg_{U_\alpha}(u)\leq 2\Delta_\alpha$. Let $v$ be the most likely neighbor of $u$ which is in $U_\alpha$.
    
    Then, continue querying the labels of one endpoint of every edge in $G$ until a vertex labeled $v$ is found. Select as candidates the original vertex (adjacent to the vertex labeled $u$) and the neighbor of the vertex labeled $v$. See also \cref{fig:matching_strategy_1}.
    \item Select two random non-adjacent vertices in $G$ as candidates.
\end{enumerate}

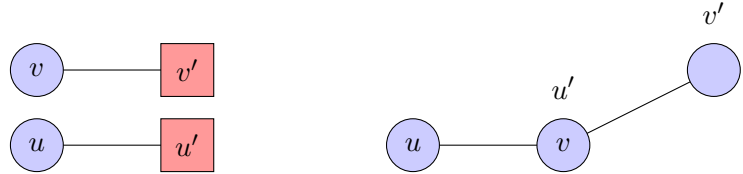
\begin{figure}[ht]
\centering
\begin{tikzpicture}[nodes={draw,circle,fill=blue!20,minimum size=20pt}]
\draw (0,0) node {$u$} -- (2,0) node [rectangle,fill=red!40] {$u'$};
\draw (0,1) node {$v$} -- (2,1) node [rectangle,fill=red!40] {$v'$};
\node [draw=none,rectangle,fill=none] at (1,-1) {(a) View of the matching $G$.};

\draw (5,0) node {$u$} -- (7,0) node [label=$u'$] {$v$} -- (9,1) node [label=$v'$] {};
\node [draw=none,rectangle,fill=none] at (7,-1) {(b) View of the PUG $U$.};
\end{tikzpicture}
\caption{Strategy (1) visually illustrated. The labels $u'$ and $v'$ of the vertices in red are the unknown labels of the selected candidates.}
\label{fig:matching_strategy_1}
\end{figure}

\paragraph{Analysis of Strategy 1:}
Before analyzing the probability of forgery using this strategy, we will show that:
\begin{enumerate}
    \itemsep 0em
    \item The condition $\deg_{U_\alpha}(u)\leq 2\Delta_\alpha$ is likely to be satisfied.
    \item The adversary is able to find a vertex with the label $v$ which is $u$'s likeliest neighbor in $U_\alpha$.
\end{enumerate}
Observe that the number of neighbors the label $u$ has in $U_\alpha$ is a random variable with expectation $\Delta_\alpha$. This is because the probability of seeing $u$ as the label of the neighbor of a random vertex in a matching is the same as the probability of seeing it as the label of a random vertex (which is $q_u$). Then, Markov's inequality implies that $\mathds{P}_u\left[\deg_{U_\alpha}(u)\geq 2\Delta_\alpha\right]\leq 1/2$ so for more than half of possible labels $u$ encountered by this strategy, it holds that $\deg_{U_\alpha}(u) < 2\Delta_\alpha$.

The adversary can repeat the initial step (henceforth ignoring the previously used edge of $G$) to find a ``good $u$" as long as it selects uniformly at random in each attempt and the distribution doesn't change when removing a small number of edges from $G$.\footnote{We assume that $G$ is large enough compared to $m$ so that the distribution over labels doesn't change in a noticeable manner when a small constant number of edges is removed from $G$. This also applies to other parts of the proof where we sample (without replacement) more than one random vertex from $G$, such as the second strategy.}

Next, we show that the adversary is able to find a vertex with the label $v$. Observe that by definition $q_v\geq \frac{\alpha}{m}$. At this step, the adversary looks at the labels of roughly half the vertices (one random endpoint of every remaining edge in $G$), stopping when $v$ is found. The probability of finding a label that is not $v$ is at most $1-\frac{\alpha}{m}$ in every edge the adversary checks. By choosing a large enough graph $G$, the adversary fails to find a vertex with label $v$ with negligible probability.

We proceed to analyze the probability of forgery. Let $u'$ and $v'$ be the (unknown) labels of the of the selected candidates. This would imply that $(u,u')$ and $(v,v')$ are edges in $U$ as the scheme only errs on non-edges. We consider the probability that $u'$ is a vertex in $U_\alpha$. By the union bound,

$$
\mathds{P}\left[ u'\in V(U_\alpha)\right]
= 1 - \sum_{w\notin V(U_\alpha)}\mathds{P}\left[ u' = w\right]
\geq 1 - \sum_{w\notin V(U_\alpha)}\frac{\alpha}{m}
\geq 1 - \alpha.
$$
Moreover, we consider the probability $u'$ is actually $v$ assuming $u'\in V(U_\alpha)$. Note that in that case, $\deg_{U_\alpha}(u)>0$ since $u'$ would be a vertex from $U_\alpha$ which is adjacent to $u$ in the PUG. Since $v$ is $u$'s likeliest neighbor in $U_\alpha$,

$$
\mathds{P}\left[u'=v \given u'\in V(U_\alpha)\right]
\geq \frac{1}{\deg_{U_\alpha}(u)}\sum_{w\in N_{U_\alpha}(u)}\mathds{P}\left[u'=w \given u'\in V(U_\alpha)\right]
= \frac{1}{\deg_{U_\alpha}(u)}.
$$
The adversary wins if $u'$ is $v$, since the labels of the candidates would be $v$ and $v'$ which are decoded as an edge. This implies the probability of forgery is at least $\frac{1-\alpha}{\deg_{U_\alpha}(u)}\geq \frac{1-\alpha}{2\Delta_\alpha}$ using this strategy.

\paragraph{Analysis of Strategy 2:}
When selecting as candidates two random (non-adjacent) vertices in $G$ at random, the adversary wins when the labels of the candidates form an edge in $U$. This is because a pair of labels adjacent in $U$ are decoded as an edge.

For each label $u\in V(U)$ we look at the probability of selecting a random vertex in $G$ which has label $u$ times the probability the second random vertex has a label from $N_U (u)$. The following calculation gives a lower bound for the probability of forgery:

\begin{align*}
\sum_{u\in V(U)}q_u\left(\sum_{v\in N_U (u)}q_v\right) &
\geq \sum_{u\in V(U)}q_u\left(\sum_{v\in N_{U_\alpha} (u)}q_v\right) \\&
\geq \sum_{u\in V(U)}q_u\left(\sum_{v\in N_{U_\alpha} (u)}\frac{\alpha}{m}\right) \\&
= \frac{\alpha}{m} \sum_{u\in V(U)}q_u\cdot \deg_{U_\alpha} (u) \\&
= \frac{\alpha}{m} \Delta_\alpha
\end{align*}

%\bigskip
\noindent
To finish, we fix $\alpha =\frac{1}{3}$. Based on the labeling scheme, select one of the two adversarial strategies depending on the value of $\Delta_{\frac{1}{3}}$ for the scheme which can be calculated in advance, without querying any of the labels. The adversary can achieve probability of forgery at least $\frac{1}{3\sqrt{m}}$ by using the first strategy when $\Delta_{\frac{1}{3}}\leq\sqrt{m}$ and the second strategy otherwise.
\end{proof}

Since the projective plane scheme needs $4 /\lceil \varepsilon^2\rceil$ symbols, we get:
%\todo{$\lceil\left(\frac{4}{\varepsilon}\right)^2\rceil$?}
\begin{corollary}
\label{cor: tightness projective}
The projective plane method is optimal (to within additive $O(1)$)  in terms of the tradeoff between the length of the label and the   probability of forgery for graphs of max-degree $1$.
\end{corollary}

\section{High Degree Graphs}\label{section:high_degree}

We now consider labeling schemes for graphs with maximal degree $d$. We first show that by adapting the ideas of labeling via coloring from \cref{coloring_paths}, one can construct a scheme with short labels that provides some weak resilience against an adaptive adversary. We show a lower bound for coloring-based schemes. 

The latter part of this section is a construction that uses the projective plane method from \cref{projective_plane_scheme} to build a scheme with labels of size $O(d)$, better than can be achieved by coloring.

\subsection{Small Labels Using Colors}\label{coloring_degree_d}
We will show that even with relatively small labels one can get {\em some} resilience against an adversary.

\begin{theorem}\label{short_labels_degree_d_scheme}
The family of graphs with maximal degree $d$ has a labeling scheme (following \cref{resilience_definition}) with size $O(\log d)$ and probability of forgery $1-\Omega (1/d)$.
\end{theorem}

We present a scheme for labeling graphs via coloring which can be thought of as an extension of the scheme from \cref{coloring_paths} for paths.

\subsubsection{Sequential Coloring}
To label a graph $G$ with maximal degree $d$, we use a palette of $3d$ colors. Order $G$ in some predetermined order on the vertices, e.g.\ the BFS order on $G$. At every step, we use a random color that is different from all previously colored neighbors. Recall that in this scheme, labels for two vertices are decoded as adjacent if and only if they have different colors.

Given any configuration of colors generated in this manner,  we would like to argue that if we know all colors other than of two vertices $x$ and $y$ which are not neighbors of each other, then the distribution of the remaining colors for $x$ is uniform from a set $S_x$ such that $2d = 3d-d \leq \lvert S_x \rvert \leq 3d$ and the same is true for $y$ and $S_y$. Then, $\left| S_x \cap S_y \right| \geq d$ so the probability for $x$ and $y$ to have the same color (and be decoded as non-adjacent) is at least $\frac{\lvert S_x\cap S_y\rvert}{\lvert S_x\rvert\lvert S_y\rvert}\geq \frac{1}{9d}$.

Unfortunately, there is a flaw in this argument, due to the distribution of the remaining colors for $x$ and $y$ in this case not being uniform. We will show that it is not far from uniform and finish the argument for the resilience of this scheme.

\bigskip
\noindent
In the following, let $\pi$ be the configuration of colors generated by this scheme which excludes at least two vertices (those are the candidates chosen by the adversary while $\pi$ consists of all the colors known to the adversary). For any vertex $x\in V(G)$, denote by $N(x)$ the neighborhood of $x$ in $G$. Also, event $\ell (x)=c$ is the event ``$x$ is colored $c$" and let $S_x$ be the set of remaining colors for $x$ assuming $\pi$.

\begin{lemma}\label{short_labels_degree_d_color_locality_lemma}
For every $x\in V(G)$ which is not part of $\pi$ and any $c\in S_x$, the probability for $\ell (x)=c$ given $\pi$ depends only on the color given to the neighbors of $x$ by $\pi$, i.e.\
$$\mathds{P}\left[
    \ell (x)=c \given \pi
\right]
= \mathds{P}\left[
    \ell (x)=c \given \pi\rvert_{N(x)}
\right]$$
\end{lemma}

\begin{proof}
In every configuration of colors generated by the scheme, the vertex $x$ is colored by a color chosen uniformly at random from the legal palette $P_{x}$ at that stage of the process which is defined by the neighbors of $x$ colored up to that point.\footnote{Not to be confused with $S_{x}$ which is the set of remaining colors for $x$ when all other colors are known (similar to the view an adversary has of $\pi$).}

Let $y_1,\ldots,y_i=x,\ldots,y_n$ be the order of $V(G)$ which produced $\pi$. The set $P_x$ is determined by the colors given to $V_0 =\{y_1,\ldots,y_{i-1}\}\cap N(x)$ and $\ell (x)$ must not be any of the colors given to $V_1 =\{y_{i+1},\ldots,y_n\}\cap N(x)$. This means that the distribution over the colors $S_x$ is a result of the colors given to $V_0$ and $V_1$ (and independent of the colors given to other vertices). Specifically, since $\pi\rvert_{V_0},\pi\rvert_{V_1}\subseteq\pi\rvert_{N(x)}\subseteq \pi$ and $c\in S_x$, the probability to see $x$ colored $c$ assuming the colors $\pi$ is the same as the probability to see the same color assuming $\pi\rvert_{N(x)}$.
\end{proof}

\begin{lemma}\label{short_labels_degree_d_color_dist_lemma}
For every $x\in V(G)$ which is not part of $\pi$ and any $c,c'\in S_x$, conditioned on the configuration of colors for the neighbors of $x$,
$$\frac{
\mathds{P}\left[
    \ell (x)=c \given \pi\rvert_{N(x)}
\right]
}{
\mathds{P}\left[
    \ell (x)=c' \given \pi\rvert_{N(x)}
\right]
} \leq \sqrt{e}$$
\end{lemma}

\begin{proof}
We analyze the distribution over $S_x$ for a vertex $x\in V(G)$ whose color is unknown to the adversary. It holds that

$$\frac{
\mathds{P}\left[
    \ell (x)=c \given \pi\rvert_{N(x)}
\right]
}{
\mathds{P}\left[
    \ell (x)=c' \given \pi\rvert_{N(x)}
\right]
}
=\frac{
\mathds{P}\left[
    \ell (x)=c \cap \pi\rvert_{N(x)}
\right]
}{
\mathds{P}\left[
    \ell (x)=c' \cap \pi\rvert_{N(x)}
\right]
}
\frac{
\mathds{P}\left[
    \ell (x)=c'
\right]
}{
\mathds{P}\left[
    \ell (x)=c
\right]
}$$
Any configuration which gives $x$ the color $c$ can be mapped to an equivalent configuration with $c$ and $c'$ switched globally. This means $\mathds{P}\left[ \ell (x)=c \right] = \mathds{P}\left[ \ell (x)=c' \right]$.

Let $d=\lvert N(x)\rvert$ and let  $z_1,\ldots,z_d$ be the vertices of $N(x)$ ordered according to the order of $V(G)$ which produced $\pi$. Recall that by the definition of the scheme, every $z_i$ was colored by a color chosen uniformly at random from the legal palette $P_{z_i}$ determined by the neighbors of $z_i$ colored up before it (this is also true regarding $x$ and $P_x$). Then,

$$\mathds{P}\left[
    \ell (x)=c \cap \pi\rvert_{N(x)}
\right]
=
\frac{1}{\lvert P_x\rvert}\cdot\prod_{i=1}^d\frac{1}{\lvert P_{z_i}\rvert}
$$
Let $k$ be the first index of a vertex colored after $x$. A key observation is that for any $i=k,\ldots,d$ the palette $P_{z_i}$ can only change by one color when switching the color of $x$. Meaning, when switching the color of $x$ from $c$ to $c'$ the most significant change is induced when $c$ is added to the palette of $z_k$ which subsequently increases all other palettes after that, i.e.

$$\frac{
\mathds{P}\left[
    \ell (x)=c \cap \pi\rvert_{N(x)}
\right]
}{
\mathds{P}\left[
    \ell (x)=c' \cap \pi\rvert_{N(x)}
\right]
}
=
\frac{
\prod_{i=k}^d\frac{1}{\lvert P_{z_i}\rvert}
}{
\prod_{i=k}^d\frac{1}{\lvert P_{z_i}\rvert + 1}
}
=
\prod_{i=k}^d\left(1+\frac{1}{\lvert P_{z_i}\rvert}\right)
\leq
\prod_{i=k}^d\left(1+\frac{1}{2d}\right)
\leq
\left(1+\frac{1}{2d}\right)^{d}
\leq \sqrt{e}
$$
This is due to the palette size never being smaller than $3d-d=2d$ (which occurs when a vertex is colored after all its neighbors) and $1\leq k\leq d$.
\end{proof}

\begin{proof}[Proof of \cref{short_labels_degree_d_scheme}]
Let $\mathcal{S}$ be the scheme described previously which uses $3d$ colors, i.e.\ the size of the labels is $\ceil{\log (3d)}$. Again, given a configuration of colors $\pi$ generated by $\mathcal{S}$, knowing all colors but two vertices $x,y$ which are not neighbors we have a distribution of remaining colors for $x$ and $y$ from sets $S_x$ and $S_y$ as described before. We also know $2d\leq\lvert S_x\rvert,\lvert S_y\rvert\leq 3d$ which means they share at least $d$ colors.

The probability that $x$ and $y$ have the same color is lowest when the colors in $S_x\cap S_y$ are the least likely to be chosen from $S_x$ and $S_y$ (since distribution isn't actually uniform) but we are able to use \cref{short_labels_degree_d_color_dist_lemma} to show a lower bound.

More specifically, since $x$ and $y$ are not neighbors, then

$$\mathds{P}\left[
    \ell (x)=\ell (y) \given \pi
\right]
= \sum_{c\in S_x\cap S_y}\mathds{P}\left[
    \ell (x)=\ell (y)=c \given \pi
\right]
= \sum_{c\in S_x\cap S_y}\mathds{P}\left[
    \ell (x)=c \given \pi
\right]\cdot \mathds{P}\left[
    \ell (y)=c \given \pi
\right].$$
Next, let $p_x=\min_{c\in S_x}\mathds{P}\left[ \ell (x)=c \given \pi \right]$ and similarly for $p_y$. By \cref{short_labels_degree_d_color_locality_lemma} when calculating $p_x$ we can look at the probabilities conditioned only on $\pi\rvert_{N(x)}$. Considering the colors in $S_x$ which have probability larger than $p_x$, by \cref{short_labels_degree_d_color_dist_lemma} their probability does not exceed $p_x\sqrt{e}$. Then,

$$p_x+\left(\lvert S_x\rvert -1\right) p_x\sqrt{e}
\geq \sum_{c\in S_x}\mathds{P}\left[ \ell (x)=c \given \pi \right] = 1$$
This means that our lower bounds are $$\mathds{P}\left[ \ell (x)=c \given \pi \right]\geq p_x\geq\frac{1}{\sqrt{e}\lvert S_x\rvert}$$ and $$\mathds{P}\left[ \ell (y)=c \given \pi \right]\geq p_y\geq\frac{1}{\sqrt{e}\lvert S_y\rvert}.$$ Finally,

$$
\mathds{P}\left[ \ell (x)=\ell (y) \given \pi \right]
\geq \sum_{c\in S_x\cap S_y} \frac{1}{\sqrt{e}\lvert S_x\rvert} \cdot \frac{1}{\sqrt{e}\lvert S_y\rvert}
= \frac{\lvert S_x\cap S_y\rvert}{e\lvert S_x \rvert\lvert S_y\rvert}
\geq \frac{d}{e \left(3d\right)^2}
= \frac{1}{9e\cdot d}
$$
This means that when using labels of size $\ceil{\log (3d)}$ the probability of forgery of this scheme is at most $1-\frac{1}{9e\cdot d}$.
\end{proof}

The next corollary follows directly from the amplification of the coloring scheme from \cref{short_labels_degree_d_scheme}, i.e.\ using it $d$ times in parallel.

\begin{corollary}
Graphs with maximal degree $d$ have a labeling scheme (following \cref{resilience_definition}) with size $O(d \log d)$ and probability of forgery which is bounded by a constant strictly less than 1.
\end{corollary}

\subsection{Covering by Matchings}\label{covering_by_matchings}

Another strategy for constructing schemes from general graphs is to use a scheme for a matching, reduce the probability of forgery if needed and combine the labels each vertex receives for any matching in which it participated.

\begin{claim}
Graphs with maximal degree $d$ have a labeling scheme (following \cref{resilience_definition}) of size $O(d\log d)$ with probability of forgery at most $1/2$.
\end{claim}

\begin{proof}
Since we already have a solution for a matching, we assume $G$ is a graph with maximal degree $d>1$. Let $p$ be the first prime such that $p+1\geq 4d$. By Vizing's theorem, all edges of the graph $G$ can be covered by $d+1$ matchings, see Misra and Gries~\cite{MisraG92}. Apply the scheme for a matching using a projective plane of order $p$ as seen in \cref{projective_plane_scheme}. For every vertex $x\in V(G)$, the label would be $\ell '(x)= \left(\ell_1 (x),\ldots,\ell_{d+1} (x)\right)$ where $\ell_i (x)$ is the label given to $x$ in the independent scheme for the $i$th matching or $\perp$ if $x$ doesn't participate in the $i$th matching. The decoder will output 1 on a pair $\ell '(x),\ell '(y)$ if and only if there exists an index $i\in\left[d+1\right]$ such that $\ell_i(x)$ and $\ell_i(y)$ are not $\perp$ and the decoder of the scheme for matching $i$ outputs 1 on the pair $\ell_i (x),\ell_i (y)$.

Consider any strategy of the adversary guessing the vertices $x,y\in V(G)$ which are not connected by an edge. There are at most $d$ matchings in which both $x$ and $y$ participate, WLOG those are the first $d$ matchings out of the $d+1$ given by Vizing's theorem. By the union bound,

$$
\mathds{P}\left[\text{forgery}\right]
\leq\sum_{i=1}^d\mathds{P}\left[\text{forgery in matching }i\right]
\leq \frac{2d}{p+1}
\leq \frac{1}{2}
$$
Generally, we must decrease the original probability of forgery to order $1/d$ by choosing $p>\Omega(d)$, or else the probability of forgery will not be bounded above by a constant. In this scheme, we achieve a probability of forgery at most $1/2$ by using labels of size $d+1$ times the original label which has size at most $\ceil{2\log (p+1)}\leq \ceil{2\log (8d)}$, i.e.\ $O(d\log d)$.
\end{proof}

\begin{claim}
\label{thm: trees}
Trees with maximal degree $d$ have a labeling scheme (following \cref{resilience_definition}) of size $O(d)$.
\end{claim}

\begin{proof}
Let $G$ be a tree with maximal degree $d$. For $\varepsilon\in \left(0,1\right)$, we construct a labeling scheme with a probability of forgery at most $\varepsilon$ by applying the projective plane method for a matching with $p$ sufficiently large to achieve a probability of forgery $\varepsilon'=\varepsilon/2$. We proceed in a manner similar to the general case, using the projective plane method for each of the $d$ matchings (since trees are Vizing class 1). Additionally, for every vertex $x\in V(G)$ let $i(x)$ be the ID of the matching the connects $x$ to its parent vertex in $G$ (for the root vertex assign any of the matching IDs for which it doesn't participate).

Under this scheme, the label for every $x\in V(G)$ is $\ell '(x)=\left(\left(\ell_1 (x),\ldots,\ell_{d+1} (x)\right),i(x)\right)$ where $\ell_j(x)$ is the label given to $x$ in the independent scheme for the $j$th matching or $\perp$ if $x$ doesn't participate in the $j$th matching. The decoder will output 1 on a pair $\ell '(x),\ell '(y)$ if and only if there exists and index $j\in\left[d\right]$ such that either $i(x)=j$ or $i(y)=j$, $\ell_j(x)$ and $\ell_j(y)$ are not $\perp$ and the decoder of the scheme for matching $j$ outputs 1 on the pair $\ell_j (x),\ell_j (y)$.

Now, let the pair $x,y\in V(G)$ be the candidates chosen by the adversary. The adversary wins if either $\dec(\ell_{i(x)}(x),\ell_{i(x)}(y))=1$ or $\dec(\ell_{i(y)}(x),\ell_{i(y)}(y))=1$ (assuming none of those values is $\perp$). The first would mean that the labels of $x$ and $y$ behave as if $y$ is the parent of $x$ and the second would mean they behave as if $x$ is the parent of $y$. On top of that, the probability of forgery in each of those cases is at most $\varepsilon'$ so the probability of forgery for the new scheme is at most $2\varepsilon'=\varepsilon$.

We conclude that there is a labeling scheme for trees of arity $d$ with label size $\ceil{2d\log (8/\varepsilon)+\log d}$ and probability of forgery at most $\varepsilon$.
\end{proof}

Our labels consist of several smaller labels (or sub-labels) which are derived from the scheme for a matching (see \cref{projective_plane_scheme}). A general difficulty in this setting is that high-degree vertices that participate in many matchings accumulate errors from the sub-labels given in each matching and therefore an adversary can use those to its advantage and increase the probability of forgery as $d$ grows. To get labels of size $O(d)$ one has to overcome this problem without scaling down the probability of forgery in the sub-labeling scheme for a matching. Trees allow for a simple solution since adjacency in trees is synonymous with having a parent-child relation and each vertex can identify the relevant sub-label for its parent.

\subsection{Labels of Size \texorpdfstring{$O(d)$}{O(d)}}

Our main objective is to construct a scheme using $O(d)$-bit labels, i.e.\ shrink the label size by a factor of $\log d$ compared to the results from \cref{covering_by_matchings}. We start by presenting a simplified version of the suggested scheme which has size $O(d\log d)$ and later explain how to improve it using a suitable data structure.

Suppose we take the graph $G$ with maximal degree $d$ and color\footnote{Unlike the coloring-based schemes suggested in \cref{coloring_paths} and \cref{coloring_degree_d}, the colors mentioned here are not used as labels themselves. They are used as a tool for the construction of a more complex scheme. In this scheme, the labels of two vertices with different colors are not necessarily decoded as having an edge.} its vertices such that no two vertices of distance 2 have the same color, i.e.\ all vertices in a neighborhood are colored using different colors. Note that every pair of colors defines an induced matching in $G$ (which we can label using a scheme for a matching).

Let $c\colon V\to \Gamma$ be a distance-2 coloring. For colors $i,j\in\Gamma$ We use $\ell_{i,j}$ (or $\ell_{j,i}$) to refer to sub-labels given by a scheme for a matching induced by vertices in $G$ colored $i$ and $j$. Then for every $x\in V(G)$ we  keep a label composed of its color and for each neighbor the color of the neighbor and the appropriate sub-label of length $O(\log(1/\varepsilon))$:
$$
\ell '(x)=
\left( c(x),
\forall y\in N_G(x)\colon \left( c(y), \ell_{c(x),c(y)} (x) \right) \right)
$$
A pair of labels $\ell '(x),\ell '(y)$ is decoded as 1 if and only if the colors from each label show up in the list of the other label and corresponding pair of sub-labels $\ell_{c(x),c(y)}(x),\ell_{c(y),c(x)}(y)$ is decoded as 1 by the decoder for this sub-label. In other words, we take from $x$'s list the sub-label corresponding to $c(y)$ and from $y$'s list the sub-label corresponding to $c(x)$ and test them.

Such distance-2 coloring is possible with $\lvert\Gamma\rvert =d^2+1$, so representing a color requires $2 \lceil \log d \rceil$ bits. This in turn means that the size of the labels is $\Theta(d\log d)$, since there are at most $d$ neighbors. Note that most of the information in these labels is not secret or random. The only part that is secret is the sub-labels and their total length is $O(d \log (1/\varepsilon))$. The main result of this section will shrink the label by utilizing a better data structure to find the correct sub-label for every relevant color, without storing the list of neighbor colors explicitly.

Dictionary data structures suitable for this task are constructed and discussed in the context of other applications in several papers such as \cite{chazelle2004bloomier,dietzfelbinger2008succinct,arbitman2010backyard}. In this section, we describe in general terms the problem solved by those constructions and adapt a specific data structure to be used in a labeling scheme for our adversarial model.

\paragraph{The Retrieval Problem:}\label{retrieval_problem}
Let $S$ be a static set of $n$ keys that come from a larger set $U$ of size $n^{O(1)}$. Let $f\colon U\mapsto\{0,1\}^r$ associate with each key an $r$-bit string. The retrieval problem is to construct a data structure that succinctly represents $f$ (on $S$) so that the value  $f(x)$ is returned when we retrieve the data associated with $x\in S$, while a query to retrieve the data associated with $x\notin S$ may return any $r$-bit string.

Usually, data structures for retrieval are used when the length $r$ of the string associated with each key is small and it is either known that only keys in $S$ will be queried or that the answers returned for keys not in $S$ do not matter. In our case, the associated strings are indeed small, sub-labels of constant size, but we {\em do} care about the results retrieved for keys {\em not} in $S$. Specifically, in our setting, we want the results to be unpredictable to the adversary, i.e.\ to have high min-entropy (e.g.\ as in a uniformly random string from some large set). Otherwise, an adversary can use this to find a vertex whose label is decoded as an edge when paired with a predictable retrieval result.

One way to obtain this requirement is by using a Bloom filter, which is a common technique to represent a set approximately. In a Bloom filter membership queries must return 1 when $x\in S$ but may err when $x\notin S$ and return 1 with probability at most $\varepsilon >0$. By adding a Bloom filter to the retrieval data structure, the scheme can actually detect most of the keys outside of $S$ and return $\perp$ for them. Any label paired with $\perp$ will not be decoded as an edge, making it useless for forgery. This actually wastes some bits, as the size of the best possible Bloom filter is roughly $n\log (1/\varepsilon)$ and we already need at least $nr$ bits to store the actual values. See also \cref{apdx:bloomier_lower_bound} which proves a general space lower bound for such constructions.

But as we will see, we can do better using the fact that from the adversary's point of view the strings in $f(S)$ have high entropy. We first prove a useful property of min-entropy and then introduce \cref{retrieval_data_structure} which uses the construction by Dietzfelbinger and Pagh~\cite{dietzfelbinger2008succinct}. By this theorem, we obtain a more succinct data structure where the result of retrieving $x\notin S$ is determined by the strings associated with the keys in $S$ (and some other random choices).

\begin{lemma}\label{min_entropy_does_not_decrease}
Let $X,Y$ be independent random variables over $\{0,1\}^r$ and let $Z=X\oplus Y$. Then
$$H_\infty [Z]\geq \max\{H_\infty [X], H_\infty [Y]\}$$
\end{lemma}

\begin{proof}
\begin{align*}
H_\infty [Z] &
\geq H_\infty \left[ Z\given X \right] \\&
= \sum_{x\in X}\mathds{P}[X=x] H_\infty \left[ Z\given X=x \right] \\&
= \sum_{x\in X}\mathds{P}[X=x] \left( -\log \max_{z\in Z} \mathds{P} \left[ Z=z \given X=x \right] \right) \\&
= \sum_{x\in X}\mathds{P}[X=x] \left( -\log \max_{y\in Y} \mathds{P} \left[ Z=x\oplus y \given X=x \right] \right) \\&
= \sum_{x\in X}\mathds{P}[X=x] \left( -\log \max_{y\in Y} \mathds{P} \left[ Y=y \given X=x \right] \right) \\&
= H_\infty \left[ Y\given X \right] \\&
= H_\infty [Y]
\end{align*}
The same argument can be used to show $H_\infty [Z]\geq H_\infty [X]$.
\end{proof}

\begin{theorem}
\label{retrieval_data_structure}
For any $r\in\mathbb{N}$, $K\geq 0$ and sufficiently large $n$,
there is a data structure of size $nr+O(\log\log n)$ bits for the retrieval problem with the following additional property: if for each $x \in S$ the string $f(x)$ is chosen independently from a distribution with min-entropy $K$, then the result of retrieving any $x\notin S$ has min-entropy at least $K$.
\end{theorem}

\begin{proof}
The retrieval data structure is based on the construction of \cite{dietzfelbinger2008succinct}. We will only show that the additional property is satisfied by the design of this retrieval data structure. The data structure of \cite{dietzfelbinger2008succinct} consists of $n$ values $a_1,\ldots,a_n\in\{0,1\}^r$ such that there is an invertible $n\times n$ binary matrix $C\in GL_n (2)$ satisfying
$$
\begin{pmatrix}
a_1\\ \vdots \\ a_n
\end{pmatrix}
= C \begin{pmatrix}
f(x_1) \\ \vdots \\ f(x_n)
\end{pmatrix}
$$
where the retrieved set is $S=\{x_1,\ldots,x_n\}$. To clarify, in this context addition is bitwise XOR of strings in $\{0,1\}^r$ and multiplication is bit-vector multiplication. Additionally, we use $k=\Theta(\log n)$ hash functions $h_1,\ldots,h_k$ which randomly map every element of $S$ to a subset of $\left[n\right]$. The construction of \cite{dietzfelbinger2008succinct} also promises subsets corresponding to elements not in $S$ are of size $k$, i.e.\ for every $x\in U$ the values $h_1(x),\ldots,h_k(x)$ do not repeat. Those are used to compute $C$ and retrieve any $x\in U$. In general, the result of retrieving $x$ is $a_{h_1(x)}\oplus a_{h_2(x)}\oplus\ldots\oplus a_{h_k(x)}$ and for $x_i\in S$ we have that $f(x_i)=a_{h_1(x_i)}\oplus a_{h_2(x_i)}\oplus\ldots\oplus a_{h_k(x_i)}$.

Consider the case where for each $x\in S$ the string $f(x)$ is independently chosen from a distribution with min-entropy $K$, e.g.\  uniformly at random from a subset of $\{0,1\}^r$ which has size $K$. For any $x\in U$, the retrieval result is
$$a_{h_1(x)}\oplus a_{h_2(x)}\oplus\ldots\oplus a_{h_k(x)}
= \bigoplus_{i=1}^k \left( C \begin{pmatrix}
f(x_1) \\ \vdots \\ f(x_n)
\end{pmatrix} \right)_{h_i (x)}
= \bigoplus_{y\in S_x\subseteq S} f\left(y\right)$$
where $S_x$ is a non empty subset of $S$ (because $C$ has full rank).

Applying \cref{min_entropy_does_not_decrease} (more than once if needed) for $\bigoplus_{y\in S_x\subseteq S} f\left(y\right)$ we conclude that the distribution of the retrieval result for $x$ has min-entropy at least $K$.
\end{proof}

We use a retrieval data structure to construct a resilient labeling scheme that associates with every neighbor the relevant sub-label without ``naming" the neighbors explicitly. This keeps the size of the labels to be $O(d)$.

\begin{theorem}\label{retrieval_degree_d_scheme}
The family of graphs with maximal degree $d$ has a labeling scheme (following \cref{resilience_definition}) of size $O(d)$ bits. More specifically, for probability of forgery $\varepsilon\in (0,1)$ the size of the label is $2d\log(1/\varepsilon)(1+o(1))$.
\end{theorem}

\begin{proof}
Let $G$ be a graph with maximal degree $d$. Let $\sigma\colon V(G)\mapsto C$ be a coloring of $G$ so that no two vertices of distance 2 have the same color. This can be done by using a proper coloring of the power graph $G^2$. Since the maximal degree of $G^2$ is at most $d^2$,
$\lvert C \rvert = d^2+1$ colors suffice. Now, each vertex in $G$ has around it vertices that are all colored using different colors.
Next, for every pair of colors $c_1,c_2\in C$ we look at the induced subgraph $G_{c_1,c_2}=G\left[\sigma^{-1}\left(\{c_1,c_2\}\right)\right]$ whose vertex set is all vertices colored $c_1$ or $c_2$. This graph is a matching since a vertex in $G$ colored $c_1$ can have at most one neighbor colored $c_2$ and vice versa. We apply the scheme based on a projective plane of order $p$ from \cref{projective_plane_scheme} to each of the $\binom{d^2+1}{2}$ matchings, giving every vertex a sub-label of size $\ceil{\log (p^2+p+1)}$ in every matching where it participates.

This scheme labels every $x\in V(G)$ by $\sigma (x)$ and an instance of the data structure promised by \cref{retrieval_data_structure} for the function $f_x\colon C\mapsto \{0,1\}^{\ceil{\log (p^2+p+1)}}$ which maps colors of neighbors of $x$ to the sub-label given to $x$ in the matching corresponding to that neighbor. Meaning, for every $y\in V(G)$ such that $(x,y)\in E(G)$ we have $f_x\left( \sigma (y)\right)=\ell_{\sigma (x),\sigma (y)} (x)$ where $\ell_{c_1,c_2}$ is the labeling function for the matching $G_{c_1,c_2}$. The subset of $C_x\subseteq C$ for which $f_x$ values are valid is the subset of colors in $x$'s neighborhood which we know has size at most $d$.

A pair of labels for vertices $x,y\in V(G)$ is decoded as 1 if and only if the pair $f_x\left( \sigma (y)\right), f_y\left( \sigma (x)\right)$ is decoded as 1 by the unique decoder corresponding to the pair of colors $\sigma(x), \sigma(y)$. There are two cases to consider:

\begin{enumerate}

\item When $\sigma(y)\in C_x$ and $\sigma(x)\in C_y$, the values $f_x\left( \sigma(y)\right)$ and $f_y\left(\sigma(x)\right)$ are indeed the sub-labels given to $x$ and $y$, respectively, in a matching in which $x$'s neighbor is colored $\sigma(y)$ and $y$'s neighbor is colored $\sigma(x)$. Here, the adversary wins if it can win against the labeling scheme used for $G_{\sigma (x),\sigma (y)}$, which has a probability of forgery at most $\frac{2}{p+1}$.
    
\item When the color $\sigma(y)\notin C_x$ (or the other way around), note that the adversary never asked for the labels of either $x$ or $y$, so the internal structure of the retrieval data structure for $x$ is unknown and $f_x\left( \sigma (y)\right)$ is the result of retrieving a value not in the valid set $C_x$. By the additional property of the retrieval data structure of \cref{retrieval_data_structure}, the distribution of this result has min-entropy at least $\log (p+1)$. This is true, since when assuming that the adversary knows the labels for the neighbors of $x$, the values stored for the keys in $C_x$ are independent uniformly random points on certain lines (corresponding to the label of each neighbor) in the projective plane of order $p$.

By the argument in \cref{projective_plane_min_entropy}, the probability of forgery in this case is determined by the min-entropy so it is still at most $\frac{2}{p+1}$.

\end{enumerate}
The label for a vertex requires $\ceil{\log (d^2+1)}$ bits to encode its color in addition to $d\ceil{\log (p^2+p+1)}+O(\log\log d)$ bits for the retrieval data structure. 
For any $\varepsilon\in (0,1)$, we choose the first prime $p+1\geq 2/\varepsilon$ and this construction yields a scheme with probability of forgery at most $\varepsilon$ and labels of size at most $\ceil{2d\log(4/ \varepsilon)+\log(d^2+1)}+O(\log\log d)$.
\end{proof}

\begin{remark}
Note that for every $x\in V(G)$ the only values of $f_x$ that matter are for the $d_x=\deg (x)\leq d$ colors of the neighbors of $x$, which means that the size of the retrieval data structure (and the label) for $x$ is proportional to the local degree of $x$ and not the global maximal degree $d$. Meaning, that the size of the label for $x$ is $2d_x\log (4/\varepsilon)+O(\log d)$.
\end{remark}

\begin{remark}
The construction outlined in \cref{retrieval_degree_d_scheme} is computationally efficient and, furthermore, can be performed in a distributed manner. Consider the two key components used in the proof: distance-2 coloring of the graph and the retrieval data structure. There are recent algorithms for distributed distance-2 coloring using $d^2+1$ colors in $O(\log n)$ rounds (and even less in the case $d\ll n$) in the CONGEST model, see Halld{\'o}rsson et al.~\cite{halldorsson2020distance,halldorsson2020coloring}. The retrieval data structure of \cite{dietzfelbinger2008succinct} is constructed locally in $poly(d)$ time for each vertex and $poly(d)\cdot n$ overall.
\end{remark}

\

Considering the scheme in \cref{retrieval_degree_d_scheme}, it seems that there should be a $\Omega (d)$ lower bound for the label size, especially when there are many vertices of degree $d$. To tackle the lower bound, one has to cautiously exclude ``easy" cases of high degree like a clique with self-loops or a collection of such cliques where we can have the labels be the clique number. Those labels have constant size independent of $d$ and no errors.
%A similar thing can be said about an independent set or a complete multipartite graph.

\section{Size Lower Bound by Learning Label Distributions}\label{section:learning}

We show a lower bound for the size of the labels in a resilient scheme by describing an adequate adversarial strategy. Generally, it can be hard to model adversarial behavior which is effective against elaborate and artificial labeling schemes. In this case, there is not much an adversary {\em can} do, which results in a natural strategy emerging. As a corollary, we will derive the following lower bound showing that for small label sizes the adversary can achieve a probability of forgery close to 1:

\begin{theorem}\label{degree_d_lower_bound}
Let $\alpha\in (0,1)$. For every $d\in\mathbb{N}$, there exists some $n>0$ such that for any graph containing $n$ non-adjacent vertices of degree at least $d$ and any labeling of that graph (following \cref{resilience_definition}) using labels of size less than $\frac{\alpha^4\cdot d}{34560}$ bits, the probability of forgery by an adversary is at least $1-\alpha$.
\end{theorem}

\begin{remark}
At first glance \cref{degree_d_lower_bound} appears stronger than what is required to show a lower bound in the setting of \cref{retrieval_degree_d_scheme} (which gives an upper bound). An alternative claim would be that to construct a resilient labeling scheme for the family of graphs with maximal degree $d$, a large label size is required. This would allow the adversary to select the hardest graph from the family to label and also solve the problem of graphs like the clique being easy to label.

Instead, we show that almost any graph in the family will suffice. In other words, resilience requires long labels {\em even if the particular graph is known when designing the labeling scheme}. In some sense, this is {\em the} interesting case, since the family of graphs with maximal degree $d$ is quite large which means the alternative formulation of the lower bound would be easy even in the non-adversarial setting. As an example, consider the trees with maximal degree $d$. Those can be labeled with error $\varepsilon$ using $2\log(1/\varepsilon )$ bits in the non-adversarial case (see \cref{apdx:degree_d_non_adversarial}). According to \cref{degree_d_lower_bound}, longer labels are required in the adversarial case.
\end{remark}

The main technique we will use in this section is based on the idea of ``If I can’t simulate what is about to happen, then I can expect to learn something by observing it" which is at the heart of the algorithm by Naor and Rothblum~\cite{naor2006learning} for learning {\em adaptively changing distributions}.

We formulate the aforementioned strategy in terms of learning the distribution of labels for the next neighbor of a center vertex in a collection of neighborhoods. The adversary essentially asks for the labels of the neighbors of a center (or many centers) one at a time and at some point decides either to guess the label of the center or simulate the distribution of labels for the next neighbor. Then, the center or the next neighbor is selected as a candidate together with a matching vertex from a different center. In some sense, the labeling of the neighborhood of each center is treated as a stochastic object the adversary is trying to imitate and when this imitation is successful, the adversary has enough vertices whose labels it hadn't queried to find a pair that has incorrect decoding with high probability. Details follow after a short overview of ACDs and known results from \cite{naor2006learning} regarding this model.

\paragraph{Adaptively Changing Distributions}
To model distributions which change adaptively over time, we use a process with some randomized mapping between states that are composed of two parts: a public state $p$ and a secret state $s$. The mapping is described by a generation algorithm $\mathcal{G}$ and sampling algorithm $\mathcal{S}$. The initial public state $p_0$ and the initial secret state $s_0$ are generated by $\mathcal{G}$ (which receives some initial input) and each new state is generated by $\mathcal{S}$ as a function of the current state and some randomness. In other words, states of the ACD are determined by $\mathcal{S}\colon S_p\times S_s\times R\mapsto S_p\times S_s$ where $S_p$ is the set of public states, $S_s$ is the set of secret states and the randomness is assumed to be taken from a set $R$.

We would like to consider the task of learning the distribution of the next public state with an algorithm that knows $\mathcal{D}$ but only has access to the sequence $p_0,p_1,p_2,\ldots$ of public states. In the context of a labeling scheme for a graph $G$, we can think of the distribution of labels for a neighbor of a vertex $x\in V(G)$ as an ACD. The generating algorithm gives $x$ a secret label $\ell(x)$ (using the graph as an auxiliary input) and the adversary can activate the sampling algorithm by asking for the label of the next neighbor of $x$ in random order. Thus, the task is to learn the distribution of labels induced by the unknown $\ell(x)$ and all previously known labels.

\begin{definition}\label{def:acd_distribution}
The distribution $D_i^s\left(p_0,\ldots,p_i\right)$ is the conditional distribution of the next public state $\mathcal{S}$ will output (in the $(i+1)$th activation) given that the past public states were $p_0,\ldots,p_i$ and given also that the initial secret generated by $\mathcal{G}$ was $s$.

For any sequence of public states $p_0,\ldots,p_i,p_{i+1}\in S_p$ and initial secret $s$, we denote by $D_i^s\left(p_0,\ldots,p_i\right)\left[p_{i+1}\right]$ the probability, conditioned on the observed public state sequence $p_0,\ldots,p_i$ and on the fact that the initial secret generated by $\mathcal{G}$ was $s$, that $\mathcal{S}$'s $(i+1)$th activation will result in $p_{i+1}$.

Generally, when the sequence of past public states is clear from the context we will write $D_i^s$. Also, when the generating algorithm $\mathcal{G}$ receives an auxiliary input, $D_i^s$ will also depend on the particular input given.
\end{definition}

\paragraph{Learning ACDs}
Give an algorithm for learning an ACD $\left(\mathcal{G},\mathcal{S}\right)$ on input length $n$, the process of learning an ACD can be described as follows:
\begin{enumerate}
    \item The generating algorithm receives input $x\in\{0,1\}^n$ and outputs the initial $s_0\in\{0,1\}^{s(n)}$ and $p_0\in S_p$. The learning algorithm receives $p_0$ and the auxiliary input $x$.\footnote{In our setting this auxiliary input will usually be the graph being labeled. Since we assume the adversary knows the graph (in fact, the adversary selects it), we will not explicitly refer to the auxiliary input.}
    \item The learning process proceeds in rounds. In each round $i$, the learning algorithm is given the new public state and may either proceed to the next round (and get the next public state by receiving the public output of $\mathcal{S}$'s next activation) or output some hypothesis $h$ as to the distribution $D_{i}^{s_0}$ of the next public state.
\end{enumerate}

\begin{theorem}[\cite{naor2006learning}]\label{acd_learning_theorem}
For any ACD $\left(\mathcal{G},\mathcal{S}\right)$ there exist an algorithm $\mathcal{A}$ and some $n_0$ such that $\mathcal{A}$ learns the ACD on inputs of size $n\geq n_0$ by activating $\mathcal{S}$ at most $O\left(\frac{s(n)}{\delta^2 (n) \varepsilon^2 (n)}\right)$ times.

Specifically, when run in the learning process described above (on input of size $n$), $\mathcal{A}$ will halt and output some hypothesis $h$ such that with probability at least $1-\delta$ (over the coin tosses of $\mathcal{G}, \mathcal{D}, \mathcal{A}$), the statistical distance between $D_i^{s_0}$ and the induced $D_h$ is $\Delta\left(D_i^{s_0}, D_h\right)\leq\varepsilon$.
\end{theorem}

\subsection{The Adversarial Strategy}\label{section:lower_bound_strategy}

The strategy we describe is executed on a set of vertices $\cal C$ and their neighbors. Recall that the main theorem requires the strategy to be valid for any graph with enough vertices of degree $d$
%, even though generally the adversarial model allows the adversary to select the worst graph for the labeling scheme. 
In our case, it is not required to select the worst graph as any graph satisfying the requirements of \cref{degree_d_lower_bound} will also contain the required structure for the execution of the strategy.

Let $G$ be a graph which contains many vertices of degree at least $d$. Find vertices $x_1,\ldots,x_n$ such that for every $i\in\left[n\right]$ the vertex $x_i$ has at least $d$ neighbors and for any $i\neq j$ the vertices $x_i$ and $x_j$ have distance at least 3.
We call ${\cal C} = \{x_1,\ldots,x_n\}$ the set of centers.

\begin{remark}
In the proof above the degree that many centers have need not be equal to the maximum degree of the graph $G$. As long as there are many centers of distance $3$ and degree at least $d$ the lower bound of $\Omega(d)$ applies.
\end{remark}

The strategy is based on the algorithm used for \cref{acd_learning_theorem}, which is applied to learning the distribution $D_i^{\ell(x)}$ of the label of the next neighbor of some $x\in\mathcal{C}$ without knowing $\ell(x)$. To clarify, the public state in this case consists of the results of all the queries asked by the adversary before selecting its candidates. This means that $D_i^{\ell(x)}$ is a conditional distribution over labels which is conditioned on the results of all previous queries, including queries related to vertices that are not adjacent to $x$.
Since the objective is actually to find suitable candidates with a limited number of queries and not simply learn the label distribution of the next neighbor, the strategy uses a modified version of the algorithm with two phases as detailed below (see also \cref{fig:degree_d_strategy}). Note that $\varepsilon,\delta\in (0,1)$ and the label size $s$ are parameters of the strategy which are also mentioned in the analysis and only fixed at a later stage.

\subsubsection{Query Phase}\label{degree_d_strategy_query_phase}

In this phase, the adversary fixes some order on the neighbors of every center in $\cal C$ and queries them until a certain condition is satisfied. We will see that it is possible for the algorithm to query some (but not all) of the neighbors of a center $x\in\mathcal{C}$, switch to the neighbors of a different center $y\in\mathcal{C}$, and later come back to query the rest of the neighbors of $x$.

\begin{definition}[Stable Center]\label{strategy_stable_condition}
Let $\mathcal{G}_i(x)$ be the conditional distribution of the label of a center $x\in\mathcal{C}$ after querying $i$ of the neighbors of $x$, given all the labels queried so far (including those not from the neighborhood of $x$).

We call $x$ {\em stable} if there is a subset of labels $L$ such that $\mathcal{G}_i(x)$ gives large weight to $L$, i.e.\ $\mathds{P}_{\mathcal{G}_i(x)}\left[\ell(x)\in L\right] \geq 1-\delta$, and for any two labels $u,v\in L$ it holds that $\Delta(D_i^u, D_i^v)\leq \varepsilon$. The distribution $D_i^u$ is the conditional distribution (as seen in \cref{def:acd_distribution}) induced by $x$ having label $u$.
\end{definition}

The querying algorithm is now: 
as long as there is some $x\in\mathcal{C}$ such that $x$ is not stable and has neighbors that were not queried:
\begin{enumerate}
    \item Let $\gamma^x_1,\ldots,\gamma^x_d\in N_G (x)$ be neighbors of $x$.
    \item Let $i$ be the position of the first neighbor of $x$ not yet queried by the algorithm. Query the vertex $\gamma^x_i$ to learn label $\ell(\gamma^x_i)$.
    \begin{enumerate}
        \item If $x$ is stable with subset $L$, then stop querying neighbors of $x$ and record the position $i$. Additionally, pick a distribution $D_x = D_i^u$ for some arbitrary $u\in L$.
        \item Otherwise, repeat with $\gamma^x_{i+1}$ if $i < d$.
    \end{enumerate}
\end{enumerate}

Note that as far as we know a center $x \in \mathcal{C}$ can become unstable after being stable, as a result of queries on vertices far from it. 
%The stopping rule is that there are no unstable nodes without neighbors that have.

\subsubsection{Decision Phase}\label{degree_d_strategy_decision_phase}

After the query phase, we can divide all centers in $\cal C$ into those that are stable and those that are fully queried, i.e.\ their $d$ neighbors mentioned in the \nameref{degree_d_strategy_query_phase} are queried. We will describe how candidates for forgery can be selected when there are enough stable centers. We will show that we are unlikely to have too many fully queried centers.

If there are at least $\left(\frac{2}{\varepsilon}\right)^{2^s}$ stable centers, then there is a pair of centers $x,y\in\mathcal{C}$ such that $j$ is the index of the last neighbor of $y$ that was queried and $\Delta(D_{x}, D_{y})\leq \varepsilon$. Select as candidates $x$ and $\gamma^y_{j+1}$. (If there are fewer stable centers, then the strategy has failed.)

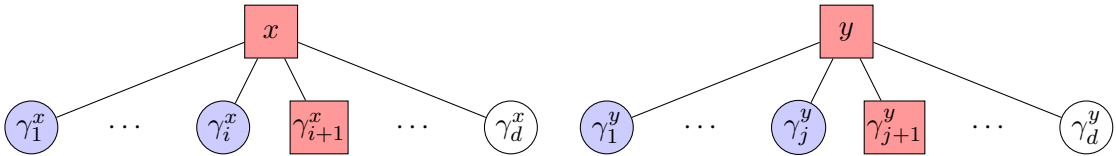
\begin{figure}[ht]
\centering
\begin{tikzpicture}[nodes={draw,circle,fill=blue!20,inner sep=1pt,minimum size=20pt},scale=.85]
\node at (0,0) [rectangle,fill=red!40] {$x$}
    child {node {$\gamma^x_1$}}
    child {node [draw=none,fill=none,minimum size=0pt] {$\cdots$} edge from parent [draw=none]}
    child {node {$\gamma^x_{i}$}}
    child {node [rectangle,fill=red!40] {$\gamma^x_{i+1}$}}
    child {node [draw=none,fill=none] {$\cdots$} edge from parent [draw=none]}
    child {node [fill=white!40] {$\gamma^x_d$}};
\node at (9,0) [rectangle,fill=red!40] {$y$}
    child {node {$\gamma^y_1$}}
    child {node [draw=none,fill=none,minimum size=0pt] {$\cdots$} edge from parent [draw=none]}
    child {node {$\gamma^y_{j}$}}
    child {node [rectangle,fill=red!40] {$\gamma^y_{j+1}$}}
    child {node [draw=none,fill=none] {$\cdots$} edge from parent [draw=none]}
    child {node [fill=white!40] {$\gamma^y_d$}};
\end{tikzpicture}
\caption{A partial view of $G$. The labels of the blue vertices are known to the adversary and the potential candidates are the red vertices. The adversary finds a pair $x,y\in\mathcal{C}$ whose distributions $D_x,D_y$ are close and the candidates are $x$ and $\gamma^y_{j+1}$. It could be the case that some centers can't be used in the \nameref{degree_d_strategy_decision_phase} because they are fully queried, but we show it is likely the adversary has enough stable centers.}
\label{fig:degree_d_strategy}
\end{figure}

\subsection{Analysis of the Strategy}
Let $\mathcal{C}_s\subseteq\mathcal{C}$ be the subset of centers that are stable at the \nameref{degree_d_strategy_query_phase}. Recall that those centers must satisfy the stability condition in $\cref{strategy_stable_condition}$ and that we assume $\lvert\mathcal{C}_s\rvert \geq \left(\frac{2}{\varepsilon}\right)^{2^s}$.

The adversary recorded distributions $\left\{D_x\right\}_{x\in\mathcal{C}_s}$ and corresponding (not queried) neighboring vertices for each distribution. We start by showing that there is a pair $x,y\in\mathcal{C}_s$ such that $\Delta\left(D_x,D_y\right)\leq\varepsilon$. Next, we show that $x$ and $\gamma^y_{j+1}$ are good candidates for forgery.

\begin{lemma}\label{packing_stable_centers}
For any $\varepsilon\in (0,1)$ and $\lvert\mathcal{C}_s\rvert \geq \left(\frac{2}{\varepsilon}\right)^{2^s}$, there exists a pair of $x,y\in\mathcal{C}_s$ such that $\Delta(D_x,D_y) \leq\varepsilon$.
\end{lemma}

\begin{proof}
We view the distributions as vectors in $\mathbb{R}^{2^s}$, i.e.\ for a center $x\in\mathcal{C}_s$ the {\em recorded distribution} $D_x$ is represented by a vector $\vec{d_x}$ where $\left(\vec{d_x}\right)_u$ is the probability of the label $u$ by $D_x$ and $\lVert\vec{d_x}\rVert_1=1$.

Let $B_{\vec{x}}^r = \{\vec{y}\in\mathbb{R}^{2^s} \suchthat \lVert \vec{x}-\vec{y}\rVert_1 \leq r \}$ be the $r$-ball around $\vec{x}$. Our distributions are points in $\partial B_{\vec{0}}^1$ which is the boundary of the unit ball (where the distance is exactly 1) and we ask how many points having statistical distance at least $\varepsilon$, which is equivalent to $L^1$ distance of at least $2\varepsilon$, are in $\partial B_{\vec{0}}^1$. We now look at all $\varepsilon$-balls around points in $\partial B_{\vec{0}}^1$, i.e.\
$$
\Phi = \bigcup_{\vec{x}\in \partial B_{\vec{0}}^1} B_{\vec{x}}^{\varepsilon}.
$$
Observe that two points in non-intersecting $\varepsilon$-balls which are packed into the volume of $\Phi$ will have statistical distance at least $\varepsilon$, since by definition the center points of those balls must have $L^1$ distance at least $2\varepsilon$. The original question is equivalent to asking for a maximal number of $\varepsilon$-balls we can pack without intersection into the volume of $\Phi$. This value is the $\varepsilon$-packing number of $\Phi$, as defined by Wainwright in~\cite[Chapter~5.1]{wainwright2019high}. It is at most

$$
\frac{\text{Vol}(\Phi)}{\text{Vol}\left(B_{\vec{0}}^{\varepsilon}\right)}
= \frac{\text{Vol}\left(B_{\vec{0}}^{1+\varepsilon}\right)-\text{Vol}\left(B_{\vec{0}}^{1-\varepsilon}\right)}{\text{Vol}\left(B_{\vec{0}}^{\varepsilon}\right)}
= \frac{\left(1+\varepsilon\right)^{2^s} - \left(1-\varepsilon\right)^{2^s}}{\left(\varepsilon\right)^{2^s}}
%= \left(\frac{1}{\varepsilon}+1\right)^{2^s} - \left(\frac{1}{\varepsilon}-1\right)^{2^s}
%\leq 2^{s+1} \left(\frac{1}{\varepsilon}+1\right)^{2^s-1}
\leq \left(\frac{2}{\varepsilon}\right)^{2^s}.
$$
Assuming the label size $s$ is fixed (depending on the maximal degree $d$), we conclude that when the number of centers $n$ exceeds the above value there must be two distributions (points) which are $\varepsilon$-close in statistical distance.
\end{proof}

\begin{lemma}\label{close_distributions_lemma}
Let $x,y\in\mathcal{C}_s$ be the selected candidates as described in the \nameref{degree_d_strategy_decision_phase}, i.e.\ both centers have neighbors which were not queried, their recorded distributions are $D_x$ and $D_y$ and $\Delta (D_x,D_y)\leq\varepsilon$. Then the probability that the pair $\ell(x)$ and $\ell(\gamma^y_{j+1})$ is {\em rejected} by the decoder is at most $2\delta+3\varepsilon$.
\end{lemma}

\begin{proof}
First, note that when $x$ is stable, then with probability at least $1-\delta$ the recorded distribution $D_x$ picked by the adversary is close to $D_i^{\ell(x)}$.  This is true, since for any $u\in L$ we have $\Delta\left(D_x, D_i^u \right)\leq \varepsilon$ and $\ell(x)$ is likely to be in $L$, i.e.\ $\mathds{P}_{\mathcal{G}_i(x)}\left[\ell(x)\in L\right]\geq 1-\delta$.

So now, for the candidate pair of centers $x,y\in\mathcal{C}_s$ picked by the adversary where  $\Delta\left(D_x,D_y\right)\leq\varepsilon$ the probability that for both of them the recorded distribution is close is at least $1-2\delta$. Assume that this is the case. 
For the candidates $x$ and $\gamma^y_{j+1}$ where $j$ is the index of last neighbor of $y$ whose label was queried by the adversary,  
using the fact that the distribution over labels for $\gamma^y_{j+1}$ is close to the distribution over labels for $x$'s next neighbor $\gamma^x_{i+1}$, we ask what is the probability (over the choice of labels, conditioned on all query results) that $\ell(x)$ rejects $\ell(\gamma^y_{j+1})$. 

Let $S \subset \{0,1\}^s$ be  the subset of labels which are rejected by $\ell(x)$, i.e.\ when decoding any label from $S$ together with $\ell(x)$ the decoder returns $0$. The probability label $\ell(\gamma^y_{j+1})$ is in $S$ is:

\begin{align}
\mathds{P} \left[ \ell(\gamma^y_{j+1}) \in S \right] &
= \mathds{P}_{u\sim D_j^{\ell(y)}} \left[ u\in S \right] \tag{a} \\&
= \mathds{P}_{u\sim D_j^{\ell(y)}} \left[ u\in S \right] - \mathds{P}_{v\sim D_i^{\ell(x)}} \left[ v\in S \right] \tag{b} \\&
\leq \max_{S'\subseteq\{0,1\}^s} \left| \mathds{P}_{u\sim D_j^{\ell(y)}} \left[ u\in S' \right] - \mathds{P}_{v\sim D_i^{\ell(x)}} \left[ v\in S' \right] \right| \nonumber \\&
= \Delta\left( D_j^{\ell(y)}, D_i^{\ell(x)} \right) \nonumber \\&
\leq \Delta\left( D_j^{\ell(y)}, D_y \right) + \Delta\left( D_y, D_x \right) + \Delta \left( D_x, D_i^{\ell(x)} \right) \nonumber \\&
\leq 3\varepsilon. \nonumber
\end{align}
%\moni{What is S' over? all subsets in $\{0,1\}^s$?}

We used in line (b) the fact that the distribution $D_i^{\ell(x)}$ must give probability zero to any label in $S$, since any possible label for $\gamma^x_{i+1}$ must be accepted by $\ell(x)$ when we condition on the center $x$ having the (true) label $\ell(x)$. Additionally, the last step uses the assumption that for the centers $x,y\in\mathcal{C}_s$ the recorded distributions are close to the ones induced by the real labels $\ell(x)$ and $\ell(y)$.

%which happens with probability at least $1-\delta$ (in each center).

The adversary fails to forge if either one of the recorded distributions is far from the one induced by the real label of the center or if $\ell(x)$ and $\ell(y_{j+1})$ are rejected by the decoder even though $D_x$ and $D_y$ are close. Using the union bound for those failure events, we conclude that the strategy fails when there are enough stable centers with probability at most $2\delta+3\varepsilon$.
\end{proof}

Next, we consider the centers whose neighborhoods were fully queried. While the analysis in the proof of \cref{close_distributions_lemma} focused on the state at the end of the \nameref{degree_d_strategy_query_phase}, the next claims will investigate what happens to the entropy of the label of a center $x\in\mathcal{C}$ each time the adversary makes a query. This will be useful to claim that the number of bad centers (which were fully queried) is not too large.

Recall the conditions for stability in \cref{strategy_stable_condition}. We show that as long as $x$ is not stable, the expected entropy of the label of $x\in\mathcal{C}$ decreases every time a neighbor of $x$ is queried. Additionally, when the algorithm queries vertices that are not neighbors of $x$ (which could make $x$ unstable again), this entropy does not increase in expectation. Towards this, when we query the $i$th neighbor of $x$ we examine two random variables:
\begin{itemize}
    \item The random variable $\tilde{\ell}(x)$ over the possible labels $\{0,1\}^s$ of $x$ which is distributed according to $\mathcal{G}_i(x)$. This is the conditional distribution of the label of $x$ given the results of all past queries. Recall that the adversary doesn't know $\ell(x)$ which is the true label of $x$ but it can generate the distribution for $\tilde{\ell}(x)$. We will argue that $H[\tilde{\ell}(x)]$ drops significantly in expectation as we query the neighbors. 
    \item The random variable $\tilde{\ell}(\gamma^x_{i+1})$ over the possible labels $\{0,1\}^s$, which corresponds to the result of the next query by the adversary (i.e.\ the label of the next neighbor of $x$), given all past public states. The distribution is produced by sampling a label $u\sim \mathcal{G}_i(x)$ and then sampling from $D_i^u$.
\end{itemize}

\begin{lemma}\label{entropy_drop_lemma}
For any center $x\in\mathcal{C}$ and distribution $\mathcal{G}_i(x)$, if $x$ is {\em not} stable after querying its $i$th neighbor, then the expected entropy drop after  querying its $(i+1)$th neighbor is bounded from below as:
$$
\mathds{E}_{u\sim \tilde{\ell}(\gamma^x_{i+1})}\left[H[\tilde{\ell}(x)]-H\left[\tilde{\ell}(x) \given \tilde{\ell}(\gamma^x_{i+1})=u\right]\right] \geq \frac{\varepsilon^2 \cdot\delta}{16}.
$$
Additionally, the expected entropy drop when querying any neighbor $\gamma^y_j$ of a center $y\neq x$ is
$$
\mathds{E}_{u\sim \tilde{\ell}(\gamma^y_{j})}\left[H[\tilde{\ell}(x)]-H\left[\tilde{\ell}(x) \given \tilde{\ell}(\gamma^y_{j})=u\right]\right] \geq 0.
$$
\end{lemma}

\begin{proof}
The first part regarding queries when $x$ is not stable is a direct application of Lemma 3.1 from \cite{naor2006learning} to our setting.

The second part regarding queries far from $x$ requires an additional argument. Since the adversary attempts to learn the distribution of the next neighbor for all centers $\mathcal{C}$, it could be the case that the algorithm alternates between querying the neighbors of some $x\in\mathcal{C}$ and querying the neighbors of another center $y\in\mathcal{C}$. We show that the entropy of $\tilde{\ell}(x)$ does not increase in expectation as a result of queries in $y$'s neighborhood, otherwise, it will be impossible to claim that the entropy of $\tilde{\ell}(x)$ is likely to decrease to zero. Fortunately, since the mutual information of two random variables is always non-negative we have that
$$
\mathds{E}_{u\sim \tilde{\ell}(\gamma^y_{j})}\left[H[\tilde{\ell}(x)]-H\left[\tilde{\ell}(x) \given \tilde{\ell}(\gamma^y_{j})=u\right]\right]
= H[\tilde{\ell}(x)]-H\left[\tilde{\ell}(x) \given \tilde{\ell}(\gamma^y_{j})\right]
= I\left(\tilde{\ell}(x);\tilde{\ell}(\gamma^y_{j})\right)
\geq 0.
$$
\end{proof}

We would like to use the above claim in order to argue that it is not likely the algorithm performs many queries on neighbors of a center $x\in\mathcal{C}$ without having the conditional entropy of the label of $x$ drop enough for the process to stop. But given the adaptive nature of the process and the fact that we only have results on the expectation of the entropy drop we have to be careful.   

For a center $x\in\mathcal{C}$, we will now look at the sequence of random variables describing the entropy of the label of $x$ (given all previous queries) throughout the random process of the \nameref{degree_d_strategy_query_phase}. We divide this process into steps which are separated by queries to neighbors of $x$, i.e.\ each step consists of queries to neighbors of centers that are not $x$ and a single query to the neighbor of $x$. This allows us to use the following result:\footnote{This is an equivalent phrasing using a sequence of random variables and their entropy instead of a probabilistic recurrence relation with a general size function. For the original formulation and proof, as well as a discussion of probabilistic recurrence relations, refer to Appendix B of \cite{naor2006learning}.}

\begin{claim}[\cite{naor2006learning}]\label{probabilistic_recurrence}
Let $X_0, X_1,\ldots X_N$ be a sequence of random variables derived from a random process with stopping time $N$.
%with the property that if for any $i$ we have that $H[X_i]=0$ then $i=N$.  %and let $N\geq 1$ be an integer valued random variable indicating the first $X_i$ with $0$ entropy, i.e.\ $H[X_N]=0$. 
Suppose that for any $i\geq 1$ in the sequence, given  the state of the random process up to the $i$th step, it holds that $\mathds{E}\left[H[X_i]\right]\leq H[X_{i-1}]-r$. Then for any $t>0$:
$$ \mathds{P}\left[N\geq t \right] \leq \frac{H[X_0]}{t\cdot r}. $$

%Let $s$ be a size function over problem instances and $T(x)=1+T(h(x))$ a probabilistic recurrence relation, where $T(0)=0$ and $h(x)$ is a random variable with $\mathds{E}\left[s(h(y))\right]\leq s(y)-r$ (for any problem instance $y$). Then for any positive integer $t$:
%$$\mathds{P}\left[T(x)\geq t\right] \leq \frac{s(x)}{t\cdot r}$$
\end{claim}

\begin{lemma}\label{bad_centers_lemma}
For any center $x\in\mathcal{C}$, the probability that after the \nameref{degree_d_strategy_query_phase} is finished all neighbors of $x$ were queried is less than $\frac{16s}{d\cdot\delta\cdot\varepsilon^2}$.
\end{lemma}

\begin{proof}
For a center $x\in\mathcal{C}$, consider the random variable, previously referred to as $\tilde{\ell}(x)$, for the label of $x$ given all previous queries throughout the \nameref{degree_d_strategy_query_phase}. Let $X_0$ be this variable at the start of this process, before any neighbor of $x$ is queried. Let $X_i$ be this variable just before $\gamma^x_i$ is queried. Recall that between queries to neighbors of $x$ other vertices which are not adjacent to $x$ can be queried by the algorithm. In the new sequence we define, a single step is between $X_i$ and $X_{i+1}$. We clarify that at every step we are looking at the random variable of the label of $x$ conditioned on all previous queries and that between $X_i$ and $X_{i+1}$ it could be the case that many vertices not adjacent to $x$ (neighbors of other centers) were queried but only a single neighbor of $x$ was queried, namely the $(i+1)$th neighbor.

Lemma~\ref{entropy_drop_lemma} gives a lower bound for the {\em expected} entropy drop in the random variable for the label of $x$ given all previous queries in the case the algorithm queries a neighbor of $x$. Additionally, the lemma  assures us that during the execution of the algorithm when vertices ``far away" from $x$ are queried, the expected entropy doesn't increase. In our sequence, a single step consists of one query to a neighbor of $x$ and some queries to neighbors of other centers. This means that we have defined a sequence where the expected entropy in each step decreases by at least $\frac{\delta\cdot\varepsilon^2}{16}$.

Even though the entropy can increase, the number of steps is unlikely to be large due to the guarantee about the expected entropy drop. Clearly, it is impossible for the entropy of any $X_i$ to be negative which limits the number of steps in the sequence. This is captured by \cref{probabilistic_recurrence} which we use to claim the probability $x$ is a bad center, i.e.\ $d$ steps were performed for $x$, is small:

$$
\mathds{P}\left[x\text{ is a bad center}\right]
%= \mathds{P}\left[\text{\# of steps}> d\right]
\leq \frac{16\cdot H[T_0]}{d\cdot\delta\cdot\varepsilon^2}
\leq \frac{16s}{d\cdot\delta\cdot\varepsilon^2}
$$
\end{proof}

To complete the proof of \cref{degree_d_lower_bound} we need to fix the desired probability of forgery, select $\delta$ and $\varepsilon$ accordingly and incorporate the assumption about the label size $s$ being small compared to the maximal degree $d$.

Fix $\alpha\in (0,1)$ and select $\delta=\alpha/12$ and $\varepsilon=\alpha/9$ for the strategy. Suppose the graph $G$ contains at least $n\geq 10\left(\frac{10}{\alpha}\right)^{2^s}$ vertices of degree $d$ and that $s\leq \frac{d\cdot\alpha^4}{34560}$. The adversary can only select candidates in the \nameref{degree_d_strategy_decision_phase} when there are at least $\left(\frac{10}{\alpha}\right)^{2^s}$ stable centers. As long as the fraction of bad centers is at most $9/10$ we will have enough stable centers. By \cref{bad_centers_lemma} and Markov's inequality, the probability we don't have enough centers is at most $\alpha/2$:

$$
\mathds{P}\left[\text{\# of bad centers}\geq\frac{9n}{10}\right]
\leq \frac{10}{9}\cdot \frac{16s}{d\cdot\delta\cdot\varepsilon^2}
= \frac{17280s}{d\cdot\alpha^3}
\leq \frac{\alpha}{2}
$$

By \cref{close_distributions_lemma}, the candidates the adversary selects fail is at most $\alpha/2$. The adversary wins when none of these events happens (not enough stable centers or bad candidates for pairing), which means that the probability of forgery is at least $1-\alpha$.

To reiterate, for any $\alpha\in (0,1)$, any $d$ and any graph containing enough vertices of degree $d$ that is labeled using labels of size less than $\frac{d\cdot\alpha^4}{34560}$, there is a strategy which will allow an adversary to forge with probability $1-\alpha$. In other words, to construct a resilient labeling scheme (even for a particular graph that is large enough), the label size has to be $\Omega (d\cdot\alpha^4)$ bits. This concludes the proof of \cref{degree_d_lower_bound}.

\subsection{Two-Sided Error}\label{lower bound for two-sided error}
The statement of \cref{degree_d_lower_bound} that we just proved is about schemes that err on one side (on non-edges). However, we note that the only part of the proof that uses this assumption is the proof of \cref{close_distributions_lemma}. Specifically between steps (a) and (b) in the calculation, we assumed that $\mathds{P}_{v\sim D_i^{\ell(x)}} \left[ v\in S \right]=0$ for the set $S$ of labels which are rejected by the decoder when paired with $\ell (x)$, which assumes that labels from $S$ cannot be assigned to a neighbor of $x$. This is not necessarily true when we have two-sided errors. 

Let $\alpha\in (0,1/2)$. We show that schemes with two-sided error have an analog of \cref{degree_d_lower_bound} with a probability of forgery at least $1/2-\alpha$ (instead of $1-\alpha$). Suppose that the scheme in question has a two-sided probability of forgery $\eta'\in (0,1/2)$. By repeating the original scheme with error $\eta'$ in parallel we construct an amplified scheme with error $\eta$ such that $\eta\leq\frac{\alpha (1-\alpha )}{4}$. Note that the label size also increased by a constant which depends logarithmicly on $\alpha$. We analyze the attack on the amplified scheme, which implies one on the original scheme.

To show that the adversarial strategy achieves a probability of forgery at least $1/2-
\alpha$, we modify the proof of \cref{close_distributions_lemma} in the following way: Note that looking at the random variable $F_x=\mathds{P}_{v\sim D_i^{\ell(x)}} \left[ v\in S \right]$ at the time of candidate selection (conditioned on all information the adversary has learned), it must be that the expectation of $F_x$ is at most $\eta$. Otherwise, there is an adversary ${\cal A}'$ that imitates  the original adversary but then selects as candidates $x$ and $\gamma_{i+1}^x$ and has a better than $\eta$ probability of winning. Therefore, applying Markov's inequality, 
$$\mathds{P}\left[F_x\geq \frac{2\eta}{1-\alpha}\right]\leq\frac{1-\alpha}{2}.$$
Assuming $F_x\leq \frac{2\eta}{1-\alpha}$, instead of subtracting $\mathds{P}_{v\sim D_i^{\ell(x)}} \left[ v\in S \right]$ in step (b) from \cref{close_distributions_lemma}, we add $\frac{2\eta}{1-\alpha}-\mathds{P}_{v\sim D_i^{\ell(x)}} \left[ v\in S \right]$ and continue in the same manner. We know this assumption does not hold with probability at most $\frac{1-\alpha}{2}$. As a result, we get that the candidates are rejected with probability at most $$2\delta+3\varepsilon+\frac{1-\alpha}{2}+\frac{2\eta}{1-\alpha}
\leq 2\delta+3\varepsilon+\frac{1}{2}.$$
The parameters $\delta$ and $\varepsilon$ are from the original description of the adversarial strategy in \cref{section:lower_bound_strategy} and appear for the same reason as in \cref{close_distributions_lemma}.

We follow the original proof, but use the modified version of \cref{close_distributions_lemma}. This includes selecting $\varepsilon$ and $\delta$ such that $2\delta+3\varepsilon=\frac{\alpha}{2}$ and showing that the probability we don't have enough centers is at most $\frac{\alpha}{2}$. As a result, for a large enough graph and label size $s\leq \Omega(d\cdot\alpha^4)$, the adversarial strategy has a probability of forgery at least $\frac{1}{2}-\alpha$.

\begin{corollary}
Let $\alpha\in (0,1/2)$.
For every $d\in\mathbb{N}$, there exists some $n>0$ such that for any graph containing $n$ non-adjacent vertices of degree at least $d$ and any labeling of that graph (following \cref{resilience_definition}) with two-sided error using labels of size less than $\Omega(\alpha^4\cdot d)$ bits, the probability of forgery by an adversary is at least $1/2-\alpha$.
\end{corollary}

Note that in the two-sided case, there is no adversarial strategy that has a probability of forgery better than $1/2$ for every scheme (so getting close to $1/2$ is the best we can do for the lower bound). Consider the following scheme for any graph that has an error of exactly $1/2$: color the graph using two colors where a color is assigned to each vertex independently at random. The decoder will return $1$ for a pair of different colors and $0$ otherwise. For a pair of vertices, the probability that the colors correctly represent whether there is an edge between them is $1/2$ and the adversary learns nothing by querying other vertices.

\section{Connections with the Simultaneous Messages Model}\label{section:smmpc}

A technique for obtaining adversarial resilient labeling that we have not explored so far in this work is to use protocols in the communication model of simultaneous messages with private coins (SMMPC) (i.e.\ {\em no} shared public randomness), as was done by Mironov, Naor and Segev~\cite{mironov2008sketching}. Recall that in the simultaneous messages model, two parties Alice and Bob are given inputs  \( x \) and \(y\) respectively and should compute some function \( f \) of the inputs \emph{without communicating with each other}. Instead, each one sends a message to a third party (a referee) who calculates \( f(x,y) \) given the messages from Alice and Bob. The quintessential example of a protocol that can save on communication (as opposed to the deterministic case) is the equality function:  Deciding whether two $m$-bits strings are equal can be done using $\Theta(\sqrt{m})$ communication (or sketches) in this model~\cite{NewmanS96,BabaiK97}.

The general approach for using the simultaneous message model with private randomness in the context of labeling is to first come up with some deterministic labeling and then show that decoder's function can be computed with low communication in that model. 
Sketches derived from protocols in this model are inherently resilient to adversaries in our setting, since  nothing is learned about the {\em private} randomness used to produce the label of a vertex by querying other vertices. 

 This allows us to have sketches for the family of trees or forests (no matter what the degree is) of length $O(\sqrt{\log n})$: the idea is to start with the deterministic labeling where every node knows its own name and the name of its parent ($\log n$ bits each) and then simply for each vertex to store an equality sketch for its own name and for its parent's name.

This, in turn, implies similar-sized sketches for low arboricity graphs (the arboricity of a graph is the minimum number of forests into which its edges can be partitioned), such as planar graphs which have arboricity 3. Therefore we get: 
\begin{corollary}
\label{cor: trees}
The family of graphs with constant arboricity and maximum degree $d$ has a labeling scheme with adversarial resilience (following \cref{resilience_definition}) of length $O(\min\{d, \sqrt{\log n}\})$ for graphs with $n$ vertices.
\end{corollary}

%In fact, any specific $n$-vertex graph $G$ where we can give deterministic labels (of size $O(\log n)$) and determine adjacency in $G$ by performing equality testing on those deterministic labels has a scheme (following \cref{resilience_definition}) using labels of size $O(\sqrt{\log n})$.

\paragraph{The Hypercube:} The same type of result is true also for specific graphs such as a collection of cliques of size $d$ and the $d$-dimensional hypercube. 
The hypercube (as a specific graph on $2^d$ vertices) has an adjacency sketch in the non-adversarial sense using constant size labels\footnote{This is implied by several previous results, but for an explicit construction refer to the note by Harms~\cite{harms2022adjacency}.} and our result (Theorem~\ref{degree_d_lower_bound}) shows that {\em this is impossible in the adversarial setting}. However, as mentioned in~\cref{packing_stable_centers}, to prove the lower bound we need the number of vertices to be double exponential in the label size, so it gives an $\Omega(\log d)$ lower bound on the label size for the hypercube.

We can show that labels of $O(\sqrt{d})$ bits suffice, using the equality protocol in the simultaneous message model with private randomness and syndrome decoding of a Hamming code, together with the fact that the equality protocol applies a linear function to the original label and sends parts of it to the referee. 

In the following section, we show that the order of $\sqrt{d}$ is actually tight for the hypercube by improving the lower bound.
The perhaps unexpected connection we show is that we can utilize techniques developed for showing lower bounds in this model in order to improve the proof of~\cref{degree_d_lower_bound}  and show that the number of centers needed is $2^{O(d^2)}$. 

\subsection{Size Lower Bound Using Characterizing Multisets}
The proof of~\cref{degree_d_lower_bound} defines the notion of a stable center and the adversary can make use of pairs of stable centers whose recorded distributions are close. Currently, for label size $s$ finding such a pair requires $\left(\frac{2}{\varepsilon}\right)^{2^s}$ centers. We propose a modification where roughly $2^{O(s^2)}$ centers will suffice.

A naive way to approach this is to modify the condition satisfied by a good pair of distributions associated with stable centers. Suppose that instead of looking for two centers $x,y$ with close distributions $D_x$ and $D_y$, the adversary looks for centers whose distributions are supported on the same subset of labels. This can be satisfied with an order of $2^s$ centers, since each label corresponds to the subset of labels it accepts (and each distribution is supported on such a subset). Unfortunately, this condition is not sufficient to prove that the decoder accepts the label of $x$ with the label of the next neighbor of $y$ (see~\cref{close_distributions_lemma}).

Instead, we suggest that the adversary look for two distributions associated with stable centers that have the same {\em characterizing multiset}. This idea is based on the technique used in the communication complexity lower bound in the randomized simultaneous messages model by Babai and Kimmel~\cite{BabaiK97} (see also Cohen and Naor~\cite{CohenN22}).

We define the notion of a characterizing multiset and analyze the alternative adversarial strategy where the modification is to select a pair of centers with similar characterizing multisets in the \nameref{degree_d_strategy_decision_phase} instead of centers with close recorded distributions. Also, we keep the meaning of the parameters $\varepsilon$ and $\delta$ from the description of the adversarial strategy in~\cref{section:lower_bound_strategy}.

\begin{definition}
For any multiset of labels $T \subset \{0,1\}^s$ such that $\lvert T\rvert =t$ and any label $u\in\{0,1\}^s$ and decoding procedure $\dec\colon \{0,1\}^s \times  \{0,1\}^s \mapsto \{0,1\}$, let the value of $T$ be  $$Q(T,u)=\frac{1}{t}\sum_{i\in \left[t\right]} \dec(T[i],u).$$
\end{definition}

\begin{definition}[Characterizing Multiset]
For any distribution $D$ over labels in $\{0,1\}^s$ we say that the the multiset $T \subset \{0,1\}^s$ {\em characterizes} $D$ wrt the decoding procedure  $\dec\colon \{0,1\}^s \times  \{0,1\}^s \mapsto \{0,1\}$ if for every label $v\in\{0,1\}^s$ it holds that
$$
\left|Q(T,v)-\mathds{P}_{u\sim D_x}\left[\dec(u,v)=1\right]\right|
\leq \frac{\varepsilon}{2}
$$
\end{definition}

We now explain the procedure of sampling a characterizing multiset for a stable center with some recorded distribution $D_x$. The adversary uses this procedure to find a characterizing multiset for every stable center (this would mean that the multiset characterizes the recorded distribution for that center) and proceeds to select candidates using a pair of centers that have the same multiset.

\begin{lemma}[Sampling a Characterizing Multiset]\label{lemma: sampling multiset}
For any stable center $x\in\mathcal{C}_s$ with recorded distribution $D_x$, let $u_1,\ldots,u_t$ be independent samples from $D_x$ where $t=2\cdot\frac{4}{\varepsilon^2}\cdot\ln{\left(\frac{2^{s+2}}{\delta}\right)}$. Then, the multiset $T=\{u_1,\ldots,u_t\}$ is a characterizing multiset for $D_x$ with probability at least $1-\frac{\delta}{2}$.
\end{lemma}

\begin{proof}
Recall that the decoder $\dec$ outputs either $0$ or $1$. As a result, for the every $i\in\left[t\right]$ and every label $v\in\{0,1\}^s$ it holds that $\mathds{E}\left[\dec(u_i,v)\right] = \mathds{P}_{u\sim D_x}\left[\dec(u,v)=1\right]$ where the randomness is over the choice $u_i\sim D_x$.

We now define random variables $X_i = \dec(u_i,v) - \mathds{P}_{u\sim D_x}\left[\dec(u,v)=1\right]$ which are independent (because the samples used to construct $T$ are independent), bounded and have expectation zero.

We can now use a Chernoff bound (see Alon and Spencer~\cite[Theorem~A.1.16]{alon2008probabilistic})
%\moni{which Chernoff? Pick one of those at the end of the Alon Spencer book.}
to show concentration for the mean of $X_i$s for a fixed label $v\in\{0,1\}^s$. More specifically,
$$
\mathds{P}\left[ \left|Q(T,v)-\mathds{P}_{u\sim D_x}\left[\dec(u,v)=1\right]\right|
> \frac{\varepsilon}{2} \right]
= \mathds{P}\left[ \left|\frac{1}{t}\sum_{i=1}^t X_i\right| > \frac{\varepsilon}{2} \right]
< 2e^{-\frac{\varepsilon^2 t}{4\cdot 2}} = \frac{\delta}{2^{s+1}}
$$
Next, by the union bound over all $v\in\{0,1\}^s$ we get that $T$ is a characterizing set with probability at least $1-\frac{\delta}{2}$.
\end{proof}

Next, we outline the different proofs for the key steps in the analysis of the adversarial strategy which are affected by the modification, i.e.\ the replacement for \cref{packing_stable_centers} and \cref{close_distributions_lemma}.

\begin{claim}
If $\lvert\mathcal{C}_s\rvert \geq O\left( 2^{\left(\frac{s}{\varepsilon}\right)^2} \right)$, then there exists a pair of $x,y\in\mathcal{C}_s$ such that the recorded distributions $D_x$ and $D_y$ have the same characterizing multiset.
\end{claim}

\begin{proof}
Recall that in~\cref{lemma: sampling multiset} the size of the multiset is $t=2\cdot\frac{4}{\varepsilon^2}\cdot\ln{\left(\frac{2^{s+2}}{\delta}\right)} = O\left(\varepsilon^{-2}\max\{s,\ln\left(\frac{1}{\delta}\right)\}\right)$.
The claim follows from counting the number of multisets of size $t$ from a universe of size $2^s$ which is
$$
\binom{2^s+t-1}{t}
\leq \frac{(2^s+t-1)^t}{t!}
= O\left( 2^{\left(\frac{s}{\varepsilon}\right)^2} \right)
.$$
\end{proof}

\begin{lemma}[Alternative Version of~\cref{close_distributions_lemma}]
Let $x,y\in\mathcal{C}_s$ be the selected candidates as described in modified the version of the \nameref{degree_d_strategy_decision_phase} using characterizing multisets, i.e.\ both centers have neighbors which were not queried and $D_x$ and $D_y$ have the same characterizing multisets. Then the probability that $\ell(x)$ and 
$\ell(\gamma^y_{j+1})$ are rejected by the decoder is at most $3\delta+3\varepsilon$.
\end{lemma}

\begin{proof}
The following is a slightly modified version of the proof of~\cref{close_distributions_lemma}. We assume that for the pair of centers $x,y\in\mathcal{C}_s$ mentioned in the \nameref{degree_d_strategy_decision_phase} the recorded distributions $D_x$ and $D_y$ are $\varepsilon$-close to the actual distributions $D_i^{\ell(x)}$ and $D_j^{\ell(y)}$. Recall that we selected $x$ and $y$ which have the same characterizing multiset so we also assume the procedure described in~\cref{lemma: sampling multiset} produced an actual characterizing multiset for both centers.

The probability the decoder $\cal D$ rejects the pair $\ell(x)$ and $\ell(\gamma^y_{j+1})$ is therefore:
\begin{align}
\mathds{P}\left[ \dec(\ell(x),\ell(\gamma^y_{j+1}))=0 \right] \nonumber &
= \mathds{P}_{u\sim D_j^{\ell(y)}}\left[ \dec(\ell(x),u)=0 \right] \nonumber \\&
= \left| \mathds{P}_{u\sim D_j^{\ell(y)}}\left[ \dec(\ell(x),u)=0 \right] - \mathds{P}_{u\sim D_i^{\ell(x)}}\left[ \dec(\ell(x),u)=0 \right] \right| \nonumber \\&
\leq \left| \mathds{P}_{u\sim D_j^{\ell(y)}}\left[ \dec(\ell(x),u)=0 \right] - \mathds{P}_{u\sim D_y}\left[ \dec(\ell(x),u)=0 \right] \right| \\&
+ \left| \mathds{P}_{u\sim D_y}\left[ \dec(\ell(x),u)=0 \right] - \mathds{P}_{u\sim D_x}\left[ \dec(\ell(x),u)=0 \right] \right| \\&
+ \left| \mathds{P}_{u\sim D_x}\left[ \dec(\ell(x),u)=0 \right] - \mathds{P}_{u\sim D_i^{\ell(x)}}\left[ \dec(\ell(x),u)=0 \right] \right| \\&
\leq 3\varepsilon \nonumber
\end{align}
We used the fact that $\mathds{P}_{u\sim D_i^{\ell(x)}}\left[ \dec(\ell(x),u)=0 \right]=0$ since the scheme only errs on non-edges. Note that terms (1) and (3) are at most $\varepsilon$ since they are bounded by the distances $\Delta(D_i^{\ell(x)},D_x)$ and $\Delta(D_j^{\ell(y)},D_y)$. 

Also, the upper bound on term (2) is due to the following observation: when two centers $x,y\in\mathcal{C}_s$ have the same characterizing multiset, for any label $v\in\{0,1\}^s$ we get that
$$
\left| \mathds{P}_{u\sim D_x}\left[\dec(u,v)=0\right]-\mathds{P}_{u\sim D_y}\left[\dec(u,v)=0\right] \right|
=\left| \mathds{P}_{u\sim D_x}\left[\dec(u,v)=1\right]-\mathds{P}_{u\sim D_y}\left[\dec(u,v)=1\right] \right|
\leq \varepsilon.
$$

Similarly to the proof of~\cref{close_distributions_lemma}, our bounds for terms (1), (2) and (3) are based on the previously mentioned assumptions, i.e.\ that the recorded distribution are close to the actual distribution and we successfully sample characterizing multisets. The adversary fails to forge if (i) At least one of those assumptions does not hold (which happens with probability at most $3\delta$ by the union bound for those failure events) or (ii) If they do hold but, the candidates end up being rejected by the decoder (which happens with probability at most $3\varepsilon$). This means that with sufficiently many stable centers the adversarial strategy fails with probability at most $3\delta+3\varepsilon$.
\end{proof}

To conclude, the technique we introduce in this section requires modifying the \nameref{degree_d_strategy_decision_phase} to sample characterizing multisets for all stable centers and select a pair that has the same multisets. Given this modifications, substituting the relevant parts from the proof of~\cref{degree_d_lower_bound} with the results of this section yields a stronger version of~\cref{degree_d_lower_bound} (with a different constant), i.e.\ the number of required vertices of large degree is $O(2^{\left(\frac{s}{\varepsilon}\right)^2})$ instead of $\left(\frac{2}{\varepsilon}\right)^{2^s}$. For the $d$-dimensional hypercube we get:

\begin{corollary}
\label{cor: hypercube}
Consider the $d$-dimensional hypercube (as a specific graph on $2^d$ vertices). It has a labeling scheme with adversarial resilience (following~\cref{resilience_definition}) using labels of size $\Theta(\sqrt{d})$ and this is in fact the best possible size.
\end{corollary}

Observe that \cref{cor: hypercube} (or even the lower bound for the label sizes for the hypercube before the improvements of this section) gives a natural counter-example in the adversarial case to the results from \cite{harms2022randomized} regarding constant size sketches for graph products, including the hypercube. Furthermore, since there is a matching upper bound for the hypercube, we conclude that we reached the asymptotically optimal size $n$ required for the lower bound of \cref{degree_d_lower_bound}.

\subsection{Separating Adversarial Resilient Sketching and the SMMPC}
At this point the reader may be wondering whether all the results regarding sketching may be phrased as ones in the SMMPC. In fact, the use of shared randomness is inherent in many labeling schemes we suggest, such as the projective plane method for labeling a matching in~\cref{projective_plane_scheme}. We show that this is necessary. 

Consider the matching graph on $2n$ vertices and consider an adversarial resilient (following~\cref{resilience_definition}) labeling scheme for this graph, applying the methods outlined in this section. In this framework we first start with some deterministic labeling and then come up with a randomized protocol in the SMMPC for the decoding function. We claim that any such protocol can be used for solving the equality function in the SMMPC: Alice and Bob, given a value $i \in [n]$, simulate the initial (deterministic) labeling $(u_i,v_i)$ given to the two nodes that form edge number $i$ in the matching, and then the message to sent the referee is simply the concatenation of the two messages $(f,g)$ sent in the randomized protocol. The referee for the equality protocol that is given a pair of messages $(f,g)$ and $(f',g')$ decides that the original values are equal if $(f,g')$ and $(f',g)$ are both accepted by the adjacency decoder for the matching graph.

We know that the equality protocol in the SMMPC requires labels of size $\Theta (\sqrt{\log n})$, so this gives us a lower bound on the length of the labeling obtained by such a method. On the other hand,  we know that the matching graph has labels of size roughly $2\log (1/\varepsilon)$ (Claim~\ref{claim: projective}). So we get that there is a large gap in this case. 

\section{Discussion, Conclusions and Open Problems}

The main construction of adversarial resilient sketches we saw used $2d \log (1/\varepsilon)(1+o(1))$-bit labels for max-degree $d$ graphs (Theorem~\ref{retrieval_degree_d_scheme}). This is only $4$ times the length of labels in the non-adversarial case (see Claim~\ref{thm: non adversarial}). One question is what exactly is the gap between these two cases. Note that 
we saw a $2 \log(1/\varepsilon)$ lower bound for the $d=1$ case (Theorem~\ref{matching_lower_bound}), but we do not have a similarly tight one for general $d$. 

We saw a specific construction for trees of max-degree $d$ (Claim~\ref{thm: trees}), but although the method is simpler, the resulting labels were of comparable size to those generated by the general construction for graphs with a maximum degree of $d$. From the lower bound (Theorem~\ref{degree_d_lower_bound}) we know that such trees need labels of size $\Omega(d)$ in order to have adversarial resilience, but perhaps there is a method with a better constant here.

\paragraph{Large Degrees:} 
We saw in Section~\ref{section:smmpc} (Corollaries~\ref{cor: trees} and \ref{cor: hypercube}) that for large $d$, say $\Theta(\log n)$, certain graphs have better than $O(d)$ adversarial resilient labeling. 
Is it possible for any specific graph $G$ to come up with a labeling scheme with similar properties (where the label size is roughly $\min(d, \sqrt{\log n})$ bits) for $G$, or are there ``rigid" graphs, say with $d$ which is $\Theta(\log n)$, that require labels of length $d$? 

\paragraph{Sketches for Other Types of Information:}
We did not consider a wide range of applications of randomness for constructing labeling schemes. For example, consider $k$-distance sketches which are used for deciding whether $\text{dist}(x,y)<k$ for vertices $x,y$ by using their labels. A construction of a $k$-distance sketch for the family of trees using $O(k)$ bits is implied by Harms~\cite{harms2020universal}. Furthermore, new results by Esperet, Harms and Kupavskii~\cite{esperet2022sketching} characterize (approximate) distance sketching in monotone graph families. In this context, we ask whether there is some general procedure for transforming such sketches to be resilient against an adaptive adversary. Our results for adjacency sketches do not seem to suggest such a procedure, so this is an intriguing open problem.

\section*{Acknowledgements}
We thank Oded Goldreich for very helpful discussions regarding amplification in the two-sided error case. 
\bibliographystyle{abbrv}
%\bibliographbystyle{alpha}
\bibliography{index}

\appendix

\section{Space Lower Bound for ``Bloomier Filters"}\label{apdx:bloomier_lower_bound}
In \cref{retrieval_problem} we use a retrieval data structure in our labeling scheme to succinctly represent sub-labels for all neighbors of a labeled vertex. Recall that a retrieval data structure answers queries $x\in S$ with values $f(x)\in \{0,1\}^r$ while for $x\notin S$ it may answer any $r$-bit string. A retrieval that approximately detects membership in $S$, i.e.\ returns $\perp$ for queries $x\notin S$ with probability at least $1-\varepsilon$, is often called a ``Bloomier filter".

Our construction and analysis can be simplified by using a ``Bloomier filter" but known data structures with this property have size $n\left(r+\log (1/\varepsilon )\right)$. It would be quite beneficial if the same could be achieved with size $n\cdot\max\{r,\log (1/\varepsilon )\}$ instead but unfortunately this is not possible. The following result is an extension of Carter et al.~\cite{carter1978exact} and \cite[{Appendix C}]{dietzfelbinger2008succinct} to data structures solving both retrieval and approximate membership.

\begin{claim}
Let $S$ be a static set of $n$ keys which come from a larger set $U$ of size $u$ and let $f\colon U\mapsto\{0,1\}^r$ be a mapping from keys to values. A retrieval data structure which detects membership in $S$ with error probability $\varepsilon\in \left(0,1\right]$ must have size at least $n\left(r+\log (1/\varepsilon )\right)-O(1)$ when $\lvert U\rvert > n^2/\varepsilon$.
\end{claim}

\begin{proof}
We show a lower bound on the size of a ``Bloomier filter" by using it to construct a data structure for the dictionary problem. In this setting, we are still considering a static set of $n$ keys $S\subseteq U$ with $r$-bit values but we require that a query on $x\in U$ reports whether or not $x\in S$ and in the positive case returns the value associated with the key $x$.

There are $\binom{\lvert U\rvert}{n}\cdot 2^{nr}$ possible dictionaries since each one is characterized by the choice of a set $S$ and the values associated with the set. Since every dictionary must have a unique memory representation, this means that every data structure for the dictionary problem must be at least $\log\binom{\lvert U\rvert}{n}+nr$ bits in the worst case.

Next, take a specific instance of a dictionary $f\colon S\mapsto\{0,1\}^r$ and let $\varepsilon\in \left(0,1\right]$. We construct a ``Bloomier filter" $\Gamma$ of optimal size which is a relaxation of $f$, i.e.\ for all $x\in S$ it answers $f(x)$ and for $x\notin S$ it answers $\perp$ with probability at least $1-\varepsilon$.

Consider the set $U_\Gamma=\left\{x\in U\mid\Gamma (x)\neq\perp\right\}$. We know that $\lvert U_\Gamma\rvert \leq \varepsilon\left(\lvert U\rvert -n\right) + n$ since $U_\Gamma$ consists of $S$ and the keys in $U\setminus S$ for which $\Gamma$ erroneously returns a value in $\{0,1\}^r$ instead of $\perp$. The fraction of those keys in $U\setminus S$ is at most $\varepsilon$. Since $\Gamma$ doesn't explicitly store the set $S$, we can add the information to turn it into a data structure for the dictionary problem. This is done by specifying a subset $S_\Gamma\subseteq U_\Gamma$ of size $n$ and testing every input $x\in U$ against this subset first: if $x\in S_\Gamma$ answer $\Gamma (x)$ and otherwise answer $\perp$. When $S_\Gamma = S$ this is a dictionary data structure for $f$.

Note that there are at least $\binom{\floor{\varepsilon\left(\lvert U\rvert-n\right)}+n}{n}$ choices for $S_\Gamma$ and each must have a unique memory representation. This means that it will require at least $\log\binom{\floor{\varepsilon\left(\lvert U\rvert-n\right)}+n}{n}$ bits in the worst case. By the lower bound on dictionary size,

$$
\lvert\Gamma \rvert + \log\binom{\floor{\varepsilon\left(\lvert U\rvert-n\right)}+n}{n} \geq \log\binom{\lvert U\rvert}{n}+nr
$$
Assuming $\lvert U\rvert > n^2/\varepsilon$, we have:

\begin{align*}
\lvert\Gamma \rvert &
\geq nr+\log\frac{\binom{\lvert U\rvert}{n}}{\binom{\floor{\varepsilon\left(\lvert U\rvert-n\right)}+n}{n}} \\&
= nr+\log\frac{\lvert U\rvert\left(\lvert U\rvert -1\right)\cdots\left(\lvert U\rvert -n+1\right)}{\left(\varepsilon\left(\lvert U\rvert-n\right)+n\right)\left(\varepsilon\left(\lvert U\rvert-n\right)+n-1\right)\cdots\left(\varepsilon\left(\lvert U\rvert-n\right)+1\right)} \\&
> nr+n\log\frac{\lvert U\rvert -n}{\varepsilon\left(\lvert U\rvert-n\right)+n} \\&
= nr+n\log\left(1/\varepsilon\right)-\log\left(1+\frac{n}{\varepsilon\left(\lvert U\rvert -n\right)}\right)^n \\&
\geq nr+n\log\left(1/\varepsilon\right)-\frac{n^2}{\varepsilon\left(\lvert U\rvert -n\right)} \\&
= nr+n\log\left(1/\varepsilon\right)-O(1)
\end{align*}
\end{proof}

\section{Graphs with Maximal Degree \texorpdfstring{$d$}{d} Without an Adversary}\label{apdx:degree_d_non_adversarial}

For completeness, we consider adjacency sketches for the family of graphs with maximal degree $d$ in the non-adversarial setting to contextualize our main results from \cref{section:high_degree,section:learning}. Let $\mathcal{F}_d$ be the family of graphs with maximal degree $d$.

In the context of deterministic labeling schemes (and universal graphs) for $\mathcal{F}_d$, we use the following theorem:

\begin{theorem}\label{deterministic_degree_d}
For every $d\geq2$, there is a deterministic labeling scheme for $\mathcal{F}_d$ using labels of size $\frac{d}{2}\log n (1+o(1))$ bits for $n$ vertex graphs. In fact, this is the best possible label size.
\end{theorem}

See Alon and Nenadov~\cite{alon2017optimal} for the optimal construction and Butler~\cite{butler2009induced} for the lower bound which is based on a counting argument for $d$-regular graphs.

\subsection{Upper Bound}
The upper bound we obtain here is shorter by roughly a factor of $4$ than the one in the adversarial case.  
\begin{claim}
\label{thm: non adversarial}
Let $\varepsilon\in (0,1)$ and $d\in\mathbb{N}$. The family $\mathcal{F}_d$ has an adjacency sketch of size  $(\lceil \frac{d}{2} \rceil +1)\lceil \log (1/\varepsilon )\rceil + \lceil 2 \log d \rceil$ with one-sided probability of error $\varepsilon$ in the non-adversarial setting.
\end{claim}

This claim is not trivial, but some variant of our construction of the proof of \cref{retrieval_degree_d_scheme} can also be used for this claim. Observe that the method for sketching a matching can be simplified when we assume there is no (adaptive) adversary. We highlight two key differences in the construction:

\begin{itemize}
\item Instead of using the projective plane method from \cref{projective_plane_scheme}, we use a password-like scheme. Meaning, give each vertex a random $\lceil \log (1/\varepsilon ) \rceil $-bit string. In the matching, a pair of vertices is decoded as being adjacent if their passwords are the same. This is unlikely to happen for any pair of non-adjacent vertices.
\item 
Roughly follow the methodology of \cref{retrieval_degree_d_scheme}: color the vertices such that no two neighbors of a vertex have the same color. One source of saving we can achieve here is storing information for only about half the neighbors: orient all the edges such that the out-degree of all vertices is at most $\lceil d/2 \rceil$.   
    Choose a random password for each vertex. Instead of storing $d$ passwords using the retrieval data structure, store only the passwords corresponding to out-edges (together with the color of the neighbor as the key). Every pair can still be decoded correctly since one of the vertices has the password for the other vertex. This requires storing the password itself plus  $\lceil {d}/{2} \rceil$ passwords of neighbors.
\end{itemize}

\subsection{Lower Bound}
\begin{claim}
Let $\varepsilon\in (0,1)$ and $d\in\mathbb{N}$. Any adjacency sketch for $\mathcal{F}_d$ which only errs on non-edges with a probability of error $\varepsilon$ in the non-adversarial setting requires $\frac{d}{4}\log (1/\varepsilon )$-bit labels.
\end{claim}

\begin{proof}
Assume towards contradiction that we have an $s$-bit one-sided adjacency sketch for $\mathcal{F}_d$ with error probability $\varepsilon$ where $s<\frac{d}{4}\log (1/\varepsilon )$. We will reduce the probability of error using amplification and show that this implies a good deterministic labeling scheme for $\mathcal{F}_d$ which contradicts the lower bound of \cref{deterministic_degree_d}.

Consider a graph $G$ on $n$ vertices in $\mathcal{F}_d$. By using this scheme $k=\frac{2\log n}{\log (1/\varepsilon )}$ times in parallel (see~\cref{amplification_method}), we reduce the probability of error to $\varepsilon^k=\frac{1}{n^2}$. By union bound,
$$
\mathds{E}\left[\#\text{ of incorrectly decoded vertex pairs}\right]
\leq \frac{1}{n^2}\cdot\binom{n}{2} < 1
$$
We conclude that there is a labeling function in the distribution over labeling functions of the amplified version of the sketch that has no errors in $G$. This gives a deterministic adjacency labeling scheme of size $s\cdot k<\frac{d}{2}\log n$ which contradicts \cref{deterministic_degree_d}.

%Assume towards contradiction that we have a $\frac{d}{20} \log (1/\varepsilon )$-bit (or shorter) adjacency sketch for $\mathcal{F}_d$. By taking $\varepsilon' = 1/n^3$, and using amplification one can show there is a sketch that has no errors (see Lemma 2.3 from \cite{harms2022randomized}). This derandomization produces a deterministic scheme with labels shorter than $\frac{d}{2} \log (n)$ bits but this is impossible by \cref{deterministic_degree_d}.\todo{change to amplification to $1/n^2$ (state one-sided error or two-sided)}
\end{proof}

\iffalse
\begin{lemma}
There exist some constant $\alpha\in (0,\frac{1}{2})$ and $n_0 (\alpha )$ such that for $n>n_0$ the number of unlabelled $d$ regular graphs on $n$ vertices is at least $2^{\alpha d n \log n}$.
\end{lemma}

This is a corollary of the main result of \cite{bollobas1982asymptotic}, showing that the number of unlabelled $d$-regular graphs for large enough $n$ is
$$
\frac{(dn)!}{2^{dn/2}\left(dn/2\right)!}\cdot\frac{e^{-\frac{d^2-1}{4}}}{(d!)^n\cdot n!}
\geq 2^{\left(\frac{d}{2}-1\right) n\log n - \frac{dn}{2}\log d-d^2}
\geq 2^{\alpha dn\log n}
$$
We conclude that the number of $n$-vertex graphs in $\mathcal{F}_d$ is at least $2^{\alpha d n \log n}$. Meaning, a deterministic adjacency labeling scheme for this family must use labels of at least $\alpha d \log n$ bits.
\fi

\begin{remark}
We emphasize that even though constructing a sketch for $\mathcal{F}_d$ requires $d$ bits, some sub-families of $\mathcal{F}_d$ like trees can be sketched using a label size which is independent of $d$ (see \cite{fraigniaud2009randomized}). This is not the case in the adversarial setting where any particular graph in $\mathcal{F}_d$ which is large enough must have label size that depends on $d$, as seen in \cref{degree_d_lower_bound}.
\end{remark}

\end{document}